\documentclass[12pt]{article}
\usepackage[papersize={8.5in,11in}, margin=3.2cm]{geometry}

\usepackage{amsmath, graphicx, appendix, amsfonts,amssymb,amsthm,paralist,subcaption,caption, color, multicol, natbib, bbm, breqn, enumerate, comment, float, tikz-network, mathtools,todonotes, setspace,wrapfig, accents, soul}
\usetikzlibrary{automata, arrows}
\usepackage[mathlines]{lineno}
\usepackage{xcolor}

\usepackage{placeins}
\usepackage{scrextend}
\usetikzlibrary{positioning}
\usepackage{hyperref}
\hypersetup{
    colorlinks=true,
    linkcolor=blue,
    filecolor=blue,      
    urlcolor=blue,
    citecolor=blue,
}
\newcommand*\patchAmsMathEnvironmentForLineno[1]{%
               \expandafter\let\csname old#1\expandafter\endcsname\csname #1\endcsname
               \expandafter\let\csname oldend#1\expandafter\endcsname\csname end#1\endcsname
               \renewenvironment{#1}%
               {\linenomath\csname old#1\endcsname}%
               {\csname oldend#1\endcsname\endlinenomath}}%
\newcommand*\patchBothAmsMathEnvironmentsForLineno[1]{%
               \patchAmsMathEnvironmentForLineno{#1}%
               \patchAmsMathEnvironmentForLineno{#1*}}%
\AtBeginDocument{%
               \patchBothAmsMathEnvironmentsForLineno{equation}%
               \patchBothAmsMathEnvironmentsForLineno{align}%
               \patchBothAmsMathEnvironmentsForLineno{flalign}%
               \patchBothAmsMathEnvironmentsForLineno{alignat}%
               \patchBothAmsMathEnvironmentsForLineno{gather}%
               \patchBothAmsMathEnvironmentsForLineno{multline}%
}

\theoremstyle{plain}
\newtheorem{proposition}{Proposition}
\newtheorem{lemma}{Lemma}
\newtheorem{corollary}{Corollary}
\newtheorem{theorem}{Theorem}

\newtheorem*{conjecture*}{Conjecture}
\newtheorem*{lemma*}{Lemma}
\theoremstyle{definition}
\newtheorem{definition}{Definition}
\newtheorem{remark}{Remark}
\newtheorem{example}{Example}
\newtheorem*{proposition*}{Proposition}

\newtheorem*{claim*}{Claim}

\newcommand{\supp}{\text{supp}}

\providecommand{\keywords}[1]{\textbf{Keywords:} #1}
\providecommand{\JEL}[1]{\textbf{JEL Classification:} #1}

\newcounter{lastnote}

\date{September 10, 2025}
\begin{document}
\font\titfont=cmr19
\title{{\titfont{Persuasion Gains and Losses from Peer Communication}}}
\author{
	Toygar T. Kerman\thanks{Institute of Economics, Corvinus University of Budapest, F\H{o}v\'{a}m t\'{e}r 8, 1093 Budapest, Hungary. email: toygar.kerman@uni-corvinus.hu. ORCID: 0000-0003-3038-3666.} \\
	\and
	Anastas P. Tenev\thanks{Institute of Economics, Corvinus University of Budapest, F\H{o}v\'{a}m t\'{e}r 8, 1093 Budapest, Hungary. email: ap.tenev@uni-corvinus.hu. ORCID: 0000-0002-2950-569X.}
    \and
    Konstantin Zabarnyi\thanks{Center for Algorithms, Data, and Market Design, Yale University, New Haven, Connecticut 06520, United States. email: konstantin.zabarnyi@yale.edu. ORCID: 0009-0006-1994-9517.\newline\indent The work was supported by the Hungarian National Research, Development and Innovation Office, Project No. K-143276, and the Center for Algorithms, Data, and Market Design at Yale~University.}
}

\maketitle

\vspace{-1.5cm}
\begin{abstract}
\noindent We study a Bayesian persuasion setting in which a sender wants to persuade a critical mass of receivers by revealing partial information about the state to them. The homogeneous binary-action receivers are located on a communication network, and each observes the private messages sent to them and their immediate neighbors.  We examine how the sender's expected utility varies with increased communication among receivers. We show that for general families of networks, extending the network can strictly benefit the sender. Thus, the sender's gain from persuasion is not monotonic in network density. Moreover, many network extensions can achieve the upper bound on the sender's expected utility among all networks, which corresponds to the payoff in an empty network. This is the case in networks reflecting a clear informational hierarchy (e.g., in global corporations), as well as in decentralized networks in which information originates from multiple sources (e.g., influencers in social media). Finally, we show that a slight modification to the structure of some of these networks precludes the possibility of such beneficial extensions. Overall, our results caution against presuming that more communication necessarily leads to better collective outcomes.
\end{abstract}

\noindent\keywords{Bayesian Persuasion; Networks; Critical Mass; Voting}

\noindent\JEL{C72, D72, D82, D85}

\section{Introduction}
The outcome of a collective choice often hinges on how information is distributed and by whom. A campaign strategist may aim to sway just enough voters to win a referendum; a CEO may try to convince key stakeholders to back a major strategic shift; or a social media platform may seek to boost user engagement to amplify an advertising campaign. Such scenarios can be studied through models in which a \emph{sender} strategically reveals information to influence the actions of multiple \emph{receivers} and affect the collective outcome (see, e.g., \citealp*{Wang13,alonso2016persuading,Taneva19,Chan19, kerman2024persuading}).

{\bf Information spillovers.} Attempts to shape collective decision making may begin by targeting individuals through carefully crafted messages. However, in most cases, people do not make decisions in isolation. Rather, they are engaged in communication in which private messages are quickly echoed or amplified through peer interaction. A message sent to one person can become part of a larger conversation when shared with one's friends/followers/acquaintances. When someone shares a campaign video or likes a product ad, their entire network is exposed to their engagement. These digital traces act as indirect endorsements that others in the network incorporate into their own decision making. In other words, a privately obtained message can often become revealed to others. The phenomenon of information spillovers described above presents a major departure from the usual treatment of Bayesian persuasion~\citep{kamenica2011bayesian} and has been largely disregarded by the literature, with only a handful of notable exceptions~\cite*{Babichenko21, liporace2021persuasion, Galperti23}. 

{\bf Our focus.} We study a persuasion environment with \textit{limited information spillovers} in which receivers are located in a fixed network. Every receiver can observe the messages sent to his immediate neighbors, and hence has access to more information than just his own message. Information is transmitted in the network only locally, which differs from \cite{Galperti23}, who assume that information spreads along all directed paths from each receiver and treat the set of feasible information structures as exogenously constrained; and from \cite{liporace2021persuasion}, who considers heterogeneous priors and incomplete network knowledge, restricting the sender to degree-based targeting of groups. In our setting, the sender fully knows the network and can target specific receivers based on their exact position. The combination of network knowledge, local spillovers, and strategic information design allows us to study realistic persuasion problems (e.g., targeted political communication or viral marketing), where the spread of information is neither global nor random. This raises the following central question: {\bf Does increased communication among receivers necessarily reduce the sender's persuasion power, or can the sender exploit it to her advantage in some networks?}

Intuitively, one might expect that giving receivers more access to others' information should make them harder to manipulate. After all, greater transparency and deliberation are often considered vital for good decision making (e.g.~\cite{angeletos2004transparency,gavazza2009transparency,de2014transparency}). However, this intuition can fail in subtle ways.\footnote{For example, \cite{levy2007decision} shows that voting in committees may result in better outcomes when the individual votes are not revealed to the public as opposed to a transparent procedure.}
The very structure of the network can create unexpected advantages for the sender. We find that under the right conditions, more dialogue among receivers can backfire on them and \textit{strictly improve} the sender's gain from persuasion. Our focus on strict improvements in the sender's value also sets the objectives and results of the current study apart from~\cite{Babichenko21} and \cite*{kerman2025bayesian}. \cite{Babichenko21} characterize a \textit{weak} order of the sender's preferences over network structures. While the nature our results is different, we make use of their sufficiency condition for achieving the upper bound on the sender's value when constructing the cases in which making the network denser by adding links benefits the sender. Unlike us,~\cite{kerman2025bayesian} focus on finding network structures that present no obstacle to the sender and never limit her persuasion abilities compared to the empty network.

{\bf Our contribution in a nutshell.} We characterize several classes of networks in which increasing communication between receivers benefits the sender and harms the receivers. In networks that naturally capture hierarchical communication between the receivers (e.g., when there are opinion leaders, or within an organizational hierarchy), the sender is not only strictly better off by adding specific links to the network, but she can even achieve the theoretical upper bound on the value (Theorem~\ref{thm:genstar}, Corollary~\ref{cor:halo}, and Proposition~\ref{pro:const}), which corresponds to the one achievable in an empty network. Moreover, the extension benefiting the sender \textit{hurts} the receivers' ability to correctly guess the outcome that is more favorable for them (Remark~\ref{rem:receivers}).

We further prove similar results on networks capturing information spillovers in \textit{decentralized} environments. Such networks can describe communication within a population divided into groups according to criteria such as political views, social media activity, and geographical location; such divisions make it harder for information to cross group boundaries. Once again, we show that there are extensions of the network beneficial to the sender (Theorem~\ref{thm:starlike} and Proposition~\ref{pro:cliques}). To complete the analysis, we discuss two cases in which the sender \textit{cannot} benefit from adding links to the network (Subsection~\ref{sub:non}). Our results provide a blueprint for approaching the problem of persuading receivers in an environment with limited information spillovers. They also act as a cautionary tale that more communication does not necessarily act as a sufficient counterbalance to a sender with an agenda. Rather, while there are cases in which communicating receivers hurt the sender's persuasion effort, we show that \emph{more interconnected sets of receivers can be more susceptible to manipulation}.

{\bf Applications.} Two prominent applications of our model that demonstrate the importance of our insights are \emph{marketing} and \emph{voting}. First, in the context of marketing, our findings imply that reaching a critical mass for a subpar product and being stuck in a ``bad'' equilibrium is likely even with consumers who ostensibly receive a lot of information.\footnote{This parallels the insight from the information cascades literature \citep*{bikhchandani1992theory,anderson1997information} that even fully rational receivers can converge to wrong outcomes despite possessing ample information. However, in our model, this situation arises from strategic information design with network spillovers rather than sequential observation of actions.} In the presence of network effects (demand-side economies of scale), it could be an uphill battle to escape this bad equilibrium, as once it is reached, the deviation of a small number of receivers is insufficient to shift the status quo.

Alternatively, our model can be interpreted as a sincere voting model with information spillovers in which the critical mass corresponds to the voting quota. In this context, our results imply that the ``wrong'' outcome from the voters' perspective is implemented with a higher probability under higher levels of political engagement of the voters. That is, more information becomes \emph{harmful} for collective decision making when the specifics of information sharing are ignored by policy makers.

\subsection{An Illustrative Example}
Let us demonstrate the possible detrimental and beneficial effects of increased communication between the receivers on the sender, starting with the former. Consider a company that promotes a new software and wishes to achieve a \textit{critical mass} of users, beyond which the software adoption will spread over time due to increased user interactions.\footnote{The adoption of products through such network effects has been extensively studied before -- see, e.g., \citet*{church1992network,katz1994systems, katona2011network}.} The quality of the product is either good $(G)$ or bad $(B)$. Concretely, suppose that there are $9$ potential clients and that $5$ constitute a critical mass. The clients initially believe that the quality is good with probability $1/3$, and they adopt the product if they believe it to be of good quality with probability at least $1/2$. The company prepares reports on the product quality, which are privately distributed among the clients.\footnote{The reports can be based on tests by the company, who decides on the specific testing policy.} The communication protocol of the company, which we call an \emph{experiment}, can be formalized by distributions $\pi(\cdot\vert G)$ and $\pi(\cdot\vert B)$ on a set of signals conditional on the quality. Let $\bar g=(g,\ldots,g)$ denote the \emph{signal} in which all receivers observe \emph{message} $g$, and define $\bar b$ analogously. Messages $g$ and $b$ can be interpreted as recommendations to adopt or reject an offer, respectively. While $\pi$ is known to the clients, under private communication they only observe their own message.

First, assume that the clients do not communicate with each other (which corresponds to the \emph{empty network}). Let $T$ be the set of signals for which 5 receivers observe $g$ and 4 receivers observe $b$; note that $\vert T\vert=\binom{9}{5}$. 
An optimal experiment for the sender, $\pi$, is given in the table in Figure~\ref{fig:1} (left), where every signal in $T$ is sent with \textit{equal} probability, namely $0.9/\vert T\vert$.\footnote{See, e.g., \cite{kerman2024persuading} for proof that the sender's expected utility in this case is~optimal.}
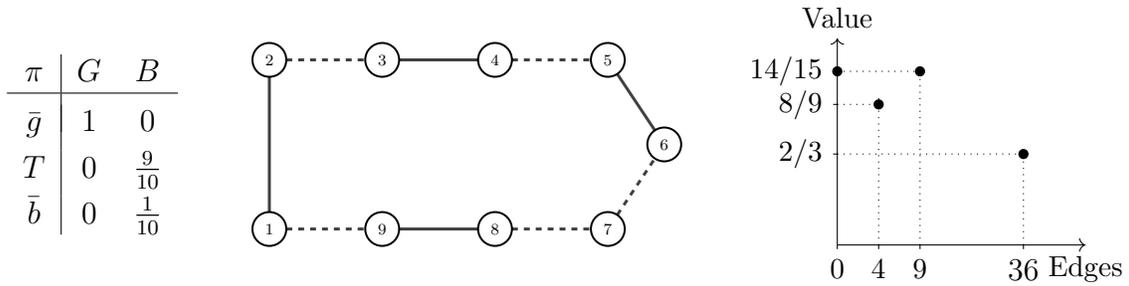
\begin{figure}[t]
\begin{minipage}{0.17\textwidth}
\begin{center}
\begin{tabular}{c|cc}
$ \pi $ & $ G $ & $ B $ \\
\hline & \\[-2ex]
$\bar g$ & 1 & 0 \\[0.1cm]
$T$ & 0 & $\frac{9}{10}$ \\[0.1cm]
$\bar b$ & 0 & $\frac{1}{10}$
\end{tabular}
\end{center}
\end{minipage}
\begin{minipage}{0.46\textwidth}
\begin{center}
\scalebox{.75}{
\begin{tikzpicture}
\Vertex[color=white,x=0,y=0,label=1]{1} 
\Vertex[color=white,x=0,y=3,label=2]{2} 
\Vertex[color=white,x=2,y=3,label=3]{3} 
\Vertex[color=white,x=2,y=0,label=9]{4}
\Vertex[color=white,x=4,y=0,label=8]{5}
\Vertex[color=white,x=6,y=0,label=7]{8}
\Vertex[color=white,x=4,y=3,label=4]{6}
\Vertex[color=white,x=6,y=3,label=5]{7}
\Vertex[color=white,x=7,y=1.5,label=6]{9}
\Edge(1)(2)
\Edge(3)(6)
\Edge(4)(5)
\Edge(7)(9)
\Edge[style=dashed](2)(3)
\Edge[style=dashed](6)(7)
\Edge[style=dashed](5)(8)
\Edge[style=dashed](1)(4)
\Edge[style=dashed](8)(9)
\end{tikzpicture}}
\end{center}
\end{minipage}
\begin{minipage}{0.35\textwidth}
\begin{center}
\begin{tikzpicture}[scale=0.55]
\draw[->] (0,0) -- (6,0) node[below] {\small Edges};
\draw[->] (0,0) -- (0,5) node[above] {\small Value};
\draw (0,4.2) -- (-0.1,4.2) node[left] {\small $14/15$};
\draw[dotted] (0,4.2) -- (2,4.2);
\draw (0,3.4) -- (-0.1,3.4) node[left] {\small $8/9$};
\draw[dotted] (0,3.4) -- (1,3.4);
\draw (0,2.2) -- (-0.1,2.2) node[left] {\small $2/3$};
\draw[dotted] (0,2.2) -- (4.5,2.2);
\draw (0,0) -- (0,-0.1) node[below] {\small 0};
\draw (1,0) -- (1,-0.1) node[below] {\small 4};
\draw[dotted] (1,0) -- (1,3.6);
\draw (2,0) -- (2,-0.1) node[below] {\small 9};
\draw[dotted] (2,0) -- (2,4.2);
\draw (4.5,0) -- (4.5,-0.1) node[below] {36};
\draw[dotted] (4.5,0) -- (4.5,2.2);
\fill (0,4.2) circle (3.5pt); 
\fill (1,3.4) circle (3.5pt); 
\fill (2,4.2) circle (3.5pt); 
\fill (4.5,2.2) circle (3.5pt); 
\end{tikzpicture}
\end{center}
\end{minipage}
\caption{The optimal experiment for fully private signaling (left); a possible network of clients (middle); and the optimal values under the empty, pair, circle, and complete networks, with respective number of edges 0, 4, 9, and 36 (right).} \label{fig:1}
\end{figure}
\noindent After observing $g$, a client believes that the quality is good with probability $(1/3\cdot 1)/(1/3\cdot 1+2/3\cdot 5/9\cdot 9/10)=1/2$.
Hence, for all realizations except $\bar b$, at least five clients adopt the product. 
Thus, the probability of reaching the critical mass -- the \emph{value} of $\pi$ -- is $14/15$. 

Consider now an alternative communication pattern between the clients. Specifically, assume that there are four communicating pairs of clients ($1$-$2$, $3$-$4$, $5$-$6$, and $8$-$9$), as shown by the solid edges in Figure~\ref{fig:1} (middle). These pairs share the information from their reports with each other before making their decisions. As a result, experiment $\pi$ is no longer optimal. Any signal for which receiver $1$ observes $b$ and receiver $2$ observes $g$ will result in \textit{both} receivers not adopting the product; in fact, receiver $2$ will deduce that the true state is $B$ from receiver $1$'s message. Since communicating pairs take the same action, optimal communication leads to reaching the critical mass with lower probability, namely $8/9$. One can check that the optimal experiment will persuade all clients to adopt with probability $1$ when the quality is $G$; and if the quality is $B$, it will persuade with probability $1/2$ exactly $5$ clients to adopt; with probability $1/3$ -- exactly $6$ clients; and with probability $1/6$ -- $0$ clients. Hence, the additional communication imposes further limitations on the company's ability to manipulate the information and makes the company worse off.

It might be natural to expect at this point that if the clients communicate even more, then the value for the sender will decrease further. A special case is when every client communicates with every other client -- i.e., they form a \emph{complete network}.
In this case, the communication effectively becomes \emph{public}.
Indeed, the value of public signaling (i.e., within each signal, all the receivers observe the same message) is the lower bound on the value of all networks (Proposition~\ref{pro:emptynetwork}).
In our example with receivers having a common prior of $1/3$, the optimal public value is $2/3$.

But what if the communication among clients is not as extreme as in the complete network? Consider extending the network of pairs given in Figure~\ref{fig:1} (middle) by the solid edges to a \emph{circle network} given by the dashed and solid edges.\footnote{The information of receiver $1$ spills over exactly to receivers $2$ and $9$, the information of receiver $2$ spills over exactly to receivers $1$ and $3$, and so on.} Surprisingly, not only does the company achieve a higher value than on the network with only the solid edges, but it can reach the empty network value $14/15$. This can be done via a simple modification to $\pi$. Specifically, take $T$ to be the set of signals for which some two consecutive receivers observe $b$ and all others observe $g$ (e.g., $(g,g,b,b,g,g,g,g,g)\in T$; note that $\vert T\vert=9$). 
As before, the interpretation of the line corresponding to $T$ in the table in Figure~\ref{fig:1} (left) is drawing a signal from $T$ uniformly at random. The definition of $T$ ensures that exactly $5$ receivers adopt the product after each realizations in $T$. It allows the company to fully exploit the power of information design on the circle network by achieving the theoretical upper bound corresponding to private signaling. Hence, the company benefits from more communication among the clients -- the value is higher on a denser network.\footnote{\emph{Density} is the ratio of the number of actual links and the number of potential links; any network obtained by adding a link to a given network is \emph{denser}.}

The example demonstrates how the change in sender's value can be non-monotonic as we move from the empty network to the other extreme case, a complete network, by adding links. On the empty network, the value is $14/15$, but when $4$ links are added to pair some of the receivers, it becomes $8/9$. When $5$ more links are formed to extend the network to a circle, which makes the network about twice as dense, the value increases again to $14/15$ -- see Figure~\ref{fig:1} (right). Finally, extending the network to a complete one drops the optimal value to $2/3$. Therefore, additional communication between receivers can be harmful for the sender, but this does not hold universally. In fact, it can also greatly help the sender. The cases in which this happens are of particular interest to us.

\subsection{Related Literature}
Our setup is closely related to the studies of \cite{Arieli19} and \cite{kerman2024persuading}. Both models study multiple receivers in a persuasion setting, but while \cite{Arieli19} characterize optimal communication for different utility functions of the sender, \cite{kerman2024persuading} focus on collective decision making following private sender-receivers communication. Unlike in our paper, in these models information is not exchanged between receivers.

\cite*{Babichenko21} study a related model, and their findings complement ours. For a given set of receivers, they provide a characterization of a partial \emph{weak} ordering over all structures of communication channels between the sender and some subsets of receivers according to the sender's \emph{universal preferences} -- i.e., robustly to the choice of utility functions and the prior distribution. The characterization relies on the notion of \emph{information domination}, where a receiver \emph{information-dominates} another one if he observes all the information channels that the other does (and possibly more). We use their result in some of our constructions of denser networks with improved sender's utility to prove that the sender can achieve a certain value in the extended network. Unlike~\cite{Babichenko21}, who focus on ordering the networks according to the sender's universal preferences and on computational aspects of their general model, we focus on the collective decision making of receivers located in a network, and we study the ability of the sender to improve her utility by changing the network structure. Importantly, we explore when an intervention in the network structure can \emph{strictly} benefit the sender, and thus their results do not readily extend to ours.

Within the same setup as ours, \cite*{kerman2025bayesian} investigate the network structures that are irrelevant to the sender's value that can be achieved. They find lower bounds on the sender's value, provide experiments achieving them, and give sufficient conditions for the optimal value to match the lower bound. In contrast, our paper investigates the network structures allowing the sender to \emph{strictly profit} from the information spillovers.

\cite{Galperti23} analyze directed networks with complete diffusion of information over all directed paths; they further assume that receivers can employ mixed strategies. This allows them to recover the revelation principle. While their model covers one extreme of information sharing in which information is a pure public good, in our case it is a \emph{local} public good. In our model, it is natural to assume pure strategies for the receivers (buyers/voters) when they face a binary decision. However, even though we study binary state and action spaces, the message space for optimal experiments may not be binary in our setting, which complicates the problem.

\cite{liporace2021persuasion} analyzes spillover effects similar to ours. Unlike us, that paper assumes that the sender only knows the distribution of the degrees of the receivers in the network, but not the full network structure. The main focus of that study is to persuade one group without dissuading another one (with different beliefs) under incomplete knowledge of the network. In contrast, we analyze how network changes affect persuasion outcomes when the sender knows the entire structure. In terms of insights, similarly to us,~\cite{liporace2021persuasion} shows that the sender can benefit from a denser network. In another model of information sharing in networks, \cite{Egorov19} consider a sender who communicates with receivers in a fixed network; a receiver either has to rely on his neighbors in the network to learn the provided information to them, or to obtain it directly from the sender for a cost. In contrast, the receivers in our model have costless access to information. In a related setup, \cite*{Candogan2020} consider threshold experiments and focus on the computational aspects of finding the sender's optimal experiment.

Our paper contributes to research on private communication and voting games. Some studies compared public and private communication under different settings \citep*{Mathevet2020,titova2022persuasion,sun2023public}, while others investigated voting games that focus on various voting rules \citep*{Bardhi18,Chan19}. We show that networks belonging to several natural classes can be slightly modified to allow the sender to achieve the upper bound on the network's value, corresponding to private persuasion.

Finally, \textit{limited} information spillovers have been studied and discovered in various contexts (the following list is not exhaustive). With a field experiment on compliance, \cite*{Drago20} show that the intensity of message sharing declines with geographic distance, and spillovers are concentrated among close neighbors. This is consistent with the findings of \cite{Marmaros06} about peer interaction of college students, where they show that people who are closer to each other communicate more. This phenomenon can occur in different contexts for different reasons; for example, \cite{gatewood1984cooperation} shows that Alaskan salmon seiners cooperate with each other despite the competition, but also only share limited information with each other because of it. \cite{ali2016ostracism} show that to ensure cooperation receivers in communities may not provide one neighbor's information to another.
Different instances of limited information spillovers happen in cases in which receivers only observe their direct neighbors' actions \citep*{corten2010co,karamched2020bayesian, mengel2024influence} or beliefs \citep{molavi2011aggregate,anunrojwong2018naive}.

{\bf Paper structure.} Section \ref{sec:not} formalizes the model setup. Section \ref{sec:pre} establishes some basic results which show the limits of the model. Section \ref{sec:applications} outlines the main results for centralized (Theorem \ref{thm:genstar}) and decentralized networks (Theorem \ref{thm:starlike}) and discusses cases in which the sender is worse off when receivers communicate more (Subsection~\ref{sub:non}). Section \ref{sec:con} concludes.

\section{Setup}\label{sec:not}
\textbf{Communication.}
Let $N=\left\{1,\ldots, n\right\}$ be the set of \emph{receivers} and $\Omega=\left\{X,Y\right\}$ be the set of \emph{states} of the world.
For any set $T$, denote by $\Delta(T)$ the set of probability distributions over $T$ with finite support.
The receivers share a common \emph{prior} belief $\lambda^0\in\Delta^\circ(\Omega)$ about the true state of the world, where $\Delta^\circ(\Omega)$ denotes the set of strictly positive probability distributions on $\Omega$. Let $S_i$ be a countable set of \emph{messages} the sender can send to receiver $i$, and let $S=\prod_{i\in N}S_i$, where the elements of $S$ are called \emph{signals}. An \emph{experiment} is a function $\pi:\Omega\rightarrow\Delta(S)$ mapping each state to a joint probability distribution over signal realizations. Let $\Pi$ be the set of all experiments.

For each signal $s\in S$, let $s_i\in S_i$ denote the message for receiver $i$. 
For each $\pi\in\Pi$, define $S^{\pi}=\supp(\pi)=\left\{s\in S|\exists\:\omega\in\Omega:\pi(s|\omega)>0\right\}$; i.e., the signals in $S$ which are sent with positive probability by $\pi$. 
Similarly, for each $i\in N$, define $S^{\pi}_i=\left\{s_i\in S_i|\exists\:\omega\in\Omega:\sum_{t\in S:t_i=s_i}\pi(t|\omega)>0\right\}$, the set of messages receiver $i$ observes with positive probability under $\pi$.

\medskip
\noindent\textbf{Preferences.}
For each $i\in N$, let $B_i=\left\{x,y\right\}$ be the set of \emph{actions} of receiver $i$. 
Let $B=\prod_{i\in N}B_i$ denote the space of action profiles. For each $i\in N$, let $u_i:B_i\times\Omega\rightarrow\left\{0,1\right\}$ be the \emph{utility function of receiver $i$}. We assume $u_i(x,X)=u_i(y,Y)=1$ and $u_i(x,Y)=u_i(y,X)=0$.
That is, the receivers want their actions to match the true state of the world. Let $Z=\left\{x,y\right\}$ be the set of \emph{outcomes}, where outcome $x$ corresponds to achieving a required \emph{critical mass} of actions $x$ and $y$ to the opposite case. Let $z^k:B\rightarrow Z$ be a map, where $z^k(a)$ is the outcome when the action profile is $a$ and the critical mass is $k$. 
Formally, 
\begin{equation*}
z^k(a)=
\begin{cases} 
     x & \text{if $|\left\{i\in N:a_i=x\right\}|\geq k$}, \\
      y & \text{otherwise}.
   \end{cases}
\end{equation*}
The \emph{sender's utility function} $v: Z \rightarrow \{0,1\}$ gets $1$ if the outcome is $x$ and $0$~otherwise.

Throughout the paper, we assume that $k\geq \lfloor \frac{n+1}{2}\rfloor$. That is, the critical mass is \emph{at least} simple majority. Moreover, we assume that $\lambda^0(X)<\frac{k}{n+k}$. 
This assumption rules out cases in which the prior is already so favorable to the sender that persuasion becomes trivial. It further prevents situations in which the sender can achieve the desired outcome with probability $1$ simply by targeting a subset of receivers.\footnote{For ease of exposition, the reader can assume that $\lambda^0(X)<1/3$; it will always fulfill the condition of preventing the sender from achieving the desired outcome with probability $1$ as long as $k$ is at least simple majority, as follows from Proposition~\ref{pro:emptynetwork} and the paragraph preceding it.}

\noindent\textbf{Information spillovers.}
An \emph{undirected network} is a map $g: N \times N \rightarrow \{0,1\}$ with $g_{ij}=g(i,j)$ and $g_{ij}=g_{ji}$, where $g_{ij}=1$ means that $i$ and $j$ are connected. 
Given a set of receivers $N$, let $G(N)$ be the set of all networks. 
We assume that receivers are in a fixed network, and each receiver in the network observes his own and his direct neighbors' message realizations. Thus, in a non-empty network, a receiver gathers more information about the true state than he would from the \emph{same} experiment under the \emph{empty} network. 
For any network $g\in G(N)$, we denote the \emph{empty} network with the same number of receivers by $g_0$.

A \emph{path} in a network $g$ between $i$ and $j$ is a sequence of players $1,\ldots,L$ s.t. $g(\ell,\ell+1)=1$ for each $\ell\in\{1,\ldots,L-1\}$ with $i=1$ and $j=L$. 
A network $g$ is \emph{connected} if for each $i,j\in N$ there exists a path in $g$ between $i$ and $j$. 
A \emph{component} of $g\in G(N)$ is a nonempty subnetwork $g'\in G(N')$ s.t. $\emptyset\neq N'\subseteq N$ and: $(i)$ $g'$ is connected; $(ii)$ if $i\in N'$ and $g_{ij}=1$, then $j\in N'$ and $g'_{ij}=1$. 

Let $N_i(g)=\left\{j\in N|g_{ij}=1\right\}\cup\{i\}$ be the \emph{neighborhood} of receiver $i$ in $g$, and let $\delta^g_i=|N_i(g)|-1$ be the \emph{degree} of $i$ in $g$. 
Let $s_i(g)=(s_j)_{j\in N_i(g)}$ be the \emph{information set} of receiver $i$ in $s$, which is the vector of messages receiver $i$ observes upon realization~$s$.

For any $g\in G(N)$, we call $g^+\in G(N)$ an \emph{extension of} $g$ (denoted $g\subsetneq g^+$) if for all $i\in N$ it holds that $N_i(g)\subseteq N_i(g^+)$ and there exists $j\in N$ s.t. $ N_j(g)\subsetneq N_j(g^+)$. 
In words, $g^+$ is a network formed by adding one or more links to $g$. Note that this does not imply that $g$ is a component of $g^+$ as they have the same number of receivers.

\begin{remark}
Importantly, the limited information spillovers that we assume are w.l.o.g. in the following sense: if the information spillovers go not just one step, but two or more on every path originating at a certain node, this can be represented by a communication network in our model that has more links. Overall, our way of representation can capture any type of (limited) undirected information spillover.
\end{remark}

Let $A^{\pi}_i(g,s)=\left\{t\in S^{\pi}\vert t_i(g)=s_i(g)\right\}$ be the \emph{association set} of receiver $i$ given $s$; i.e., the set of signals $i$ considers possible upon realization $s$. 
For any $g\in G(N)$, $\pi\in\Pi$, and $s\in S^{\pi}$, the \emph{posterior belief vector} $\lambda^{s,g}\in\Delta(\Omega)^n$ is defined by
\begin{equation*}
\lambda^{s,g}_i(\omega)=\frac{\sum_{t\in A^{\pi}_i(g,s)}\pi(t|\omega)\lambda^0(\omega)}{\sum_{\omega'\in\Omega}\sum_{t\in A^{\pi}_i(g,s)}\pi(t|\omega')\lambda^0(\omega')},\quad i\in N, \omega\in\Omega\label{equ:BU}.
\end{equation*}

\noindent That is, $\lambda^{s,g}_i(\omega)$ is receiver $i$'s Bayesian-updated belief that the state is $\omega$ upon observing $s_i(g)$. For any $g\in G(N)$, $\pi\in\Pi$, and $i\in N$, let $S^{\pi}_i(g)=\prod_{j\in N_i(g)}S^{\pi}_j$ be the space of message vectors that $i$ can observe. Due to their utility functions, the receivers choose the action corresponding to the more likely state given their posterior. Formally, the \emph{strategy} of receiver $i$ is given by $\alpha^{\pi,g}_i:S^{\pi}_i(g)\rightarrow B_i$ s.t. for any $s\in S^{\pi}$:\footnote{As is standard in the Bayesian persuasion literature, we break ties in the sender's favor.}

\begin{equation*}
\alpha^{\pi,g}_i\left(s_i(g)\right)=
\begin{cases} 
     x & \text{if $\lambda^{s,g}_i(X)\geq\frac{1}{2}$}, \\
      y & \text{otherwise}\label{equ:SB}.
   \end{cases}
\end{equation*}

\noindent Define the set of signals that achieve the critical mass on $g$ under $\pi$ as $Z^g_x(\pi)=\left\{s\in S^{\pi}\vert z^k\left(\alpha^{\pi,g}(s)\right)=x\right\}$. Let $a\in B$ be an action profile and $z=z^k(a)$ be an outcome. The \emph{value} of an experiment $\pi\in\Pi$ for a critical mass $k$ is defined as the sender's expected utility under $\pi$ on the network $g$. 
As we fix $\lambda^0$ and $\alpha^{\pi,g}$ throughout the paper, we write $V^{\pi}_k(g)=V^{\pi}_k(\lambda^0,g,\alpha^{\pi,g})$, where:
\begin{align*}
V^{\pi}_k(g)=\mathbb{E}_{\lambda^0}\left[\mathbb E_{\pi}\left[v(z^k\left(\alpha^{\pi,g}\left(s\right)\right)\right]\right]&=\lambda^0(X)\sum_{s\in Z^g_x(\pi)}\pi(s\vert X)+\lambda^0(Y)\sum_{s\in Z^g_x(\pi)}\pi(s\vert Y).
\end{align*} 

\noindent That is, given $n$, $k$, and $g$, the value of an experiment is equal to the probability of reaching the critical mass. 
An experiment $\pi^*\in\Pi$ is \emph{optimal on} $g$ for critical mass $k$ if
$V^{\pi^*}_k(g)= \sup_{\pi\in\Pi}V^{\pi}_k(g)$. We further call $\sup_{\pi\in\Pi}V^{\pi}_k(g)$ \emph{the value of the network} $g$.

In the illustrative example, we observed that adding links to the network could first lower and then raise the sender’s value, showing that $V^{\pi^*}_k(\cdot)$ needs not to move in a single direction as the network becomes denser. In particular, the example exhibited a violation of what we shall call \emph{monotonicity} (in network density) -- i.e., the property that the network value weakly decreases whenever the network is extended. 

\begin{definition}[Monotonicity]
Let $\left(g^{(t)}\right)^m_{t=0}$ be a finite sequence of networks s.t. for each $t\in\{0,\ldots,m-1\}$, $g^{(t)}\subsetneq g^{(t+1)}$.
The sequence $\left(g^{(t)}\right)^m_{t=0}$ is \emph{monotonic} if $V^{\pi^*}_k\left(g^{(t)}\right)\ge V^{\pi^*}_k\left(g^{(t+1)}\right)$, for all $t\in\{0,\ldots,m-1\}$.
\end{definition}

One can construct monotonic sequences of networks, e.g., by removing the circle and/or the network of pairs in the illustrative example.
Nevertheless, some sequences of networks are non-monotonic; we examine these in this paper.

\section{Basic Results}\label{sec:pre}
On the empty network, there is a straightforward optimal experiment à la~\cite{kamenica2011bayesian}.\footnote{An experiment is \emph{straightforward} if for all $i\in N$ it holds that ($i$) $S^\pi_i\subseteq B_i$ and ($ii$) for all $g\in G(N)$ and $s\in S^\pi$ with $s_i=a_i$, $\alpha^{\pi,g}_i(s_i(g))=a_i$.} Moreover, this optimal experiment sends $x$ to all the receivers with probability $1$ if the state is $X$, and to a set of $k$ receivers (selected randomly with equal probability) if the state is $Y$. Note that this experiment is anonymous -- i.e., permuting the receivers does not change the distribution of signals. Introducing limited information spillovers makes searching for optimal experiments much more conceptually challenging. In particular, neither of the characteristics of optimal experiments on the empty network outlined in the previous paragraph holds in a general network. Straightforwardness does not hold, since a receiver computes his posterior after observing multiple signals that might be shared with other receivers. Anonymity trivially does not hold due to the possibly intricate structure of the network.
Finally, the sender might be able to benefit from garbling information in state $X$, and thus sending the same message in $X$ with probability $1$ is not necessarily optimal (see Example 1 in \cite{kerman2025bayesian}).

While the information that a receiver gathers on a non-empty network $g$ can always be replicated on $g_0$, the converse is not necessarily true. This implies that the upper bound of the sender's gain from persuasion is the optimal value on the empty network, which we denote by $V^n_k$. Since our model boils down to the setup of \cite{kerman2024persuading} when the network is empty, we conclude that for $\lambda^0(X)<\frac{k}{n+k}$, $V^n_k = \frac{n+k}{k}\lambda^0\left(X\right)$.\footnote{The result of \cite{kerman2024persuading} can also be deduced from \cite{Arieli19}.}
On the other hand, the optimal \emph{public} experiment (i.e., all receivers observe the same message within every signal) is \emph{independent} of the network structure, and thus guarantees the sender a lower bound, which we denote by $V^p$. In particular, it always yields the same value $V^p=2\lambda^0(X)$ for any $k$, as either all receivers are persuaded or none are.\footnote{Since receivers share a common prior, the situation is equivalent to persuading a single receiver, as in \cite{kamenica2011bayesian}.} 
Note that $V^p<V^n_k$ for any $k<n$.

\begin{proposition}[Bounds]\label{pro:emptynetwork}
The value of an optimal experiment lies within $\left[V^p,V^n_k\right]$. 
\end{proposition}

\noindent All proofs from the main text can be found in Appendix~\ref{ap:proofs}.

Thus, if we start from the empty network and extend it to the complete network by a consecutive addition of links, the optimal value must strictly decrease after some extension. Proposition \ref{pro:emptynetwork} implies that if the receivers could form links without cost, they would form a complete network.\footnote{This is because the ability of the sender to manipulate the receivers into taking action $x$ when the state is $Y$ is the lowest possible for public signaling.} However, surprisingly, we show that forming a \textit{limited} number of links (i.e., not all possible ones) might not be in their best interest.

Let us return to our illustrative example in which the sender benefits from the creation of multiple links, and therefore the change in the optimal value is non-monotonic. 
To systematically explain this result, we make two observations.
First, if two receivers have exactly the same neighborhood, then they can be treated identically.
This feature, which we call \emph{symmetry}, holds for every network $g$.

\begin{lemma}[Symmetry]\label{lem:symmetry}
For every experiment on some network $g$, there is an experiment with the same value s.t. any two receivers with the same neighborhood get the same message with probability $1$.
\end{lemma}

That is, if two receivers observe identical sets of messages, the sender cannot influence them differently in a meaningful way. Since their information sets are indistinguishable, they form same posteriors and behave identically; treating them symmetrically does not restrict the sender’s ability to~persuade.

Our second observation is that despite the presence of communication, the sender can achieve the upper bound $V^n_k$ on a circle.\footnote{A connected network $g\in G(N)$ is a \emph{circle} if $\delta^g_i=2$ for all $i\in N$.} Intuitively, this is because the circle limits the spread of information in a controlled way: no receiver observes more than two others, and the network lacks central or highly-connected nodes that would dominate the information flow.
Hence, while limited information spillovers on a circle decrease the variety of receiver groups that the sender can target and persuade in state $Y$ (i.e., minimal winning coalitions) relative to the empty network, the sender retains enough flexibility to differentiate messages that reach a critical mass.

\begin{lemma}[Circle] \label{lem:circle}
The sender can achieve $V^n_k$ on a circle network whenever $n>3$. 
\end{lemma}

\noindent We can now state our first result on the non-monotonicity of the network value. 
This result establishes that this phenomenon is not an idiosyncrasy of our illustrative example. Moreover, it helps to set the groundwork for the more general results we discuss in Section \ref{sec:applications}. Its discussion will further demonstrate important differences between our setup and other models.

\begin{proposition}[Pairs]\label{pro:pairs}
If all receivers in a network with $n>3$ have degree at most $1$, then either the sender can achieve $V^n_k$ or there exists an extension that strictly benefits the sender.\footnote{The logic of the proposition easily applies to networks that consist of connected triples and at most one singleton.}  
\end{proposition}
The network from Proposition~\ref{pro:pairs} consists of pairs and singletons. Although a network of isolated pairs is stylized, it can represent a market with individuals receiving information privately and having a limited opportunity to compare notes. Such customers may be reached through direct sales calls or personalized ads; they may discuss products only with a close contact, like a friend or partner.
In contrast, a circle represents a slightly more interconnected environment (such as a neighborhood or a small online community) in which each customer sees or hears about the product from two others, but broader visibility is still constrained. While one might expect the more connected network to make customers less susceptible to persuasion, our result shows that the sender may actually prefer a circle.

Proposition~\ref{pro:pairs} helps explaining the observation in our illustrative example regarding the non-monotonicity of the network value. On the one hand, the existence of multiple connected pairs limits the sender's possible choice of experiments by Lemma~\ref{lem:symmetry}, since the pairs will take the same action regardless of the messages. This implies that the optimal value is likely to be less than $V^n_k$. On the other hand, extending this network to a circle allows the sender to exploit private communication and induce different actions. We prove Proposition~\ref{pro:pairs} by establishing the existence of networks on which the sender cannot achieve $V^n_k$, and thus extending them to a circle is a strict improvement to the value. Consequently, by Lemma~\ref{lem:circle}, the sender can achieve $V^n_k$ and benefit from a denser network.

The network of pairs can also be treated as a weighted voting setup, or as a situation in which the customers have non-unitary demand for a product.\footnote{See \cite{kerman2024information} for a detailed analysis of information design with weighted voting.} Specifically, one may think of voter blocs that vote for the same alternative \citep*{gormley2008exploring,eguia2011voting,cooperman2024bloc,grimmer2025measuring}, and customers within the same segment that make the same purchase decisions \citep*{athanassopoulos2000customer,broderick2007behavioural}. Moreover, all voters within a bloc (or all customers within a segment) have the same ex ante belief, as well as the same ex post belief. This makes it more difficult for the sender to understand how to optimally reach a critical mass. However, in a circle network, voter blocs (or customer segments) are not as strictly delineated, and thus the sender might induce different ex post beliefs and achieve a higher probability of persuasion.

Note that Lemma~\ref{lem:circle} also follows from the results of~\cite{Babichenko21}, as they explicitly mention.\footnote{They refer to this lemma appearing in an early working paper version~\cite{kerman2021persuading}.} To explain this, we first state their definition of information domination, adapted to our notations.  
\begin{definition}[Information domination]\label{def:dominating}
Given a network $g\in G(N)$, an ordered pair of receivers $(i,j)$ is \emph{information-dominating} if $N_j(g)\subseteq N_i(g)$.
In this case, we also say that receiver $i$ \emph{information-dominates} receiver $j$.
\end{definition}

\noindent That is, a receiver information-dominates another if he observes at least the same messages. The authors prove that \emph{if there are no information-dominating pairs in a network, the sender can achieve the upper bound on her value (corresponding to private signaling)}. It is easy to see that, indeed, a circle has no information-dominating pairs. We present Lemma~\ref{lem:circle} as a simple illustration to the non-monotonicity of the value.
Additionally, our constructive proof of this lemma has the following two advantages over the approach of~\cite{Babichenko21}:
\begin{inparaenum}[(i)]
    \item we use only two messages, while they use a number of messages that depends on the number of receivers;
    \item we use a simple construction of the optimal experiment that mirrors the approach on the empty network, while they rely on secret sharing protocols from cryptography.
\end{inparaenum} 

The analysis of the circle network shows that our approach and the approach of~\cite{Babichenko21} in achieving the upper bound $V^n_k$ may overlap. In general, however, they are \textit{complementary} in determining whether this bound can be reached. More precisely, there are networks for which we cannot construct the optimal experiment analytically, yet we can still conclude that the sender can achieve $V^n_k$ via the absence of information-dominating pairs.
Conversely, there are networks that \emph{do} contain information-dominating pairs and fall outside the scope of their analysis, but for which we can construct experiments that achieve $V^n_k$ (see Example~\ref{ex:starlike}). Together, the two approaches provide a more comprehensive understanding of when and how the sender can fully exploit information design with information spillovers.

In addition to our constructive approach, we shall also employ the information-domination argument to establish non-monotonicity results in other classes of networks. The following lemma, which is an immediate corollary of Theorem 3.5 in \cite{Babichenko21}, will help us to formalize the use of this approach; it provides a sufficient condition under which extending the network benefits the sender, namely breaking all information-dominating pairs.

\begin{lemma}\label{cor:info}
Let $g\in G(N)$ and suppose that for any $\pi\in\Pi$ it holds that $V^\pi_k(g)<V^n_k$. 
If $g'\in N$ with $g\subsetneq g'$ has no information-dominating pairs, then there exists $\pi'\in\Pi$ s.t. $V^{\pi'}_k(g')=V^n_k>V^\pi_k(g)$.
\end{lemma}

\begin{remark}
    \label{rem:dominating}
    In addition, if a network is modified so that no existing information-dominating pair is removed, then the sender's value cannot increase. Although this result is not directly used in any of our proofs, it provides an intuition on the type of network extensions that may benefit the sender.
\end{remark}

\section{Persuasion Gains from Network Extensions}\label{sec:applications}
Section \ref{sec:pre} established the existence of cases in which receivers are less likely to uncover the truth when they share information with one another. In this section, we extend our analysis and explore a variety of network topologies to show that under certain conditions, the sender can \textit{strictly} benefit from communicating with receivers who are in a denser network. That is, we characterize different classes of networks for which extending the network creates a non-monotonicity in the way the sender's gain from persuasion changes with network density. This non-monotonicity hinges not only on how much communication occurs, but also on how it takes place. Hierarchical influence, local echo chambers, and fragmented communities can create different constraints and opportunities for the sender. By analyzing these structures, we identify conditions that determine whether persuasion power is enhanced. Our results show that the sender can often benefit \textit{not despite, but because} of increased communication, as long as the network retains certain asymmetries or local bottlenecks in how information spreads. Moreover, Remark~\ref{rem:receivers} implies that such sender-beneficial extensions are detrimental for the receivers. In Subsection~\ref{sub:non}, we briefly explore networks in which the sender cannot benefit from adding links.

\subsection{Network Extensions Beneficial for the Sender}

In many real-world settings, communication does not flow symmetrically among all individuals, but follows a hierarchical pattern. Certain receivers act as local hubs of information (due to status \citep*{dong2015inferring}, visibility \citep*{jankowski2016picture,benevento2025impact}, or structural roles \citep*{Burt92, Burt04}), while others have more limited access to the broader information environment. These hierarchies are especially common in social media, organizations, and group deliberation contexts, where a small number of participants shape the way information is interpreted or redistributed.

One network that captures such a structure is a \emph{star}, in which a \emph{central} node is connected to all others (\emph{peripheral} nodes), who do not communicate among themselves \citep*{bala2000noncooperative,goeree2009search,galeotti2010law,nora2024exploiting}. This network is too restrictive for our purposes, as it excludes both the possibility of multiple centers and any communication among peripheral receivers, which are often observed in real-world networks. We study broader classes of network structures that capture real-world communication more accurately, allowing for hierarchical information flow beyond what the star structure permits.

We define several classes of networks that reflect these more general structures. We introduce \emph{stellar} networks (Definition \ref{def:stellar}; see also Figure \ref{fig:stellar}) that mirror the dynamics around influencers or public figures on platforms like $\mathbb X$ or Instagram: a central user broadcasts content to many followers, who observe their posts, but might also observe each other's posts. \emph{Constellations} (Definition \ref{def:constellation}, example in Figure \ref{fig:halo} -- right) generalize this structure, capturing environments with multiple local opinion leaders (such as online communities organized around Subreddit moderators, Facebook group admins, or niche content creators), each of whom information-dominates a subset. 

\emph{Galaxy} networks (Definition \ref{def:galaxy}) go further than constellations, representing disconnected clusters of communication, where each cluster evolves largely in isolation. These can correspond to separate discussion spaces, private chat groups, or isolated communities within larger platforms in which little cross-talk occurs between clusters. 
These can be interest groups which espouse opposing viewpoints or support opposing political parties.
These same structural patterns might also appear outside of social media -- e.g., in organizational settings, such as corporate hierarchies or teams operating in different areas with a central coordinator, which exhibit low inter-unit communication~\citep*{hansen1999search,mathiesen2010organizational}.

We begin by examining stellar networks, which capture hierarchical communication structures s.t. a central node observes or influences all others within the network. While these networks include a clear hub (as in star networks), we allow communication among peripheral receivers as well, as long as the resulting information flow preserves a hierarchical structure.\footnote{While we introduce the graph structures as \emph{networks}, we shall also consider networks that include these structures as \emph{components}, in which case we shall refer to them accordingly (e.g., a \emph{stellar component}).}

\begin{definition} \label{def:stellar}
A \emph{stellar network} is a network $g$ that either:
\begin{inparaenum}[(i)]
    \item consists of a single node; or
    \item contains a node $r$ that is connected to all the other nodes in $g$, these other nodes are partitioned to \emph{at least two} (disjoint and nonempty) subsets each of which is itself a stellar network, and no edges connect nodes in different partition~elements.
    \end{inparaenum}
\end{definition}

Note that a stellar network $g$ has a corresponding directed graph $T^g$ representing the partial ordering induced on the receivers by the information-domination relation, and $T^g$ is a directed tree. Let the \emph{depth} of a node in a stellar network be the number of edges in the unique path from the root to that node in $T^g$. In particular, note that there are at least two nodes of each depth $\ell>0$. Define the \emph{depth} of a stellar network $g$ to be the maximum depth of a node in $g$. For instance, a star (with at least two peripheral nodes) is a stellar network of depth $1$.

\begin{example}[A stellar network's directed tree] Consider the stellar network with $n=6$ in Figure~\ref{fig:stellar}.
Notice that node $4$ information-dominates $1, 2, 3, 5$, and $6$; and node $6$ information-dominates $3$ and $5$. 
This is reflected in the directed tree in \ref{fig:tree} representing the information-domination relation.

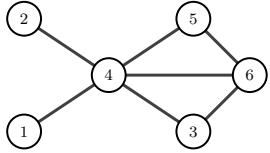
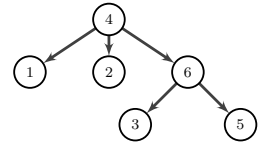
\begin{figure}[H]
\begin{subfigure}[t]{0.5\linewidth}
\begin{center}
\scalebox{.75}{
\begin{tikzpicture}
\Vertex[color=white,opacity=0.8,x=0,y=0,label=1]{1} 
\Vertex[color=white,opacity=0.8,x=0,y=2,label=2]{2} 
\Vertex[color=white,x=3,y=2,label=5]{5} 
\Vertex[color=white,opacity=0.8,x=3,y=0,label=3]{3} 
\Vertex[color=white,opacity=0.8,x=1.5,y=1,label=4]{4}
\Vertex[color=white,x=4,y=1,label=6]{6}
\Edge(1)(4)
\Edge(4)(5)
\Edge(2)(4)
\Edge(3)(4)
\Edge(4)(6)
\Edge(3)(6)
\Edge(6)(5)
\end{tikzpicture}}
\end{center} \caption{A stellar network} \label{fig:stellar}
\end{subfigure} 
\begin{subfigure}[t]{0.5\linewidth}
\begin{center}
\scalebox{.70}{
\begin{tikzpicture}
\Vertex[color=white,x=1.5,y=2,label=4]{4} 
\Vertex[color=white,x=1.5,y=1,label=2]{2} 
\Vertex[color=white,opacity=0.8,x=0,y=1,label=1]{1} 
\Vertex[color=white,opacity=0.8,x=3,y=1,label=6]{6}
\Vertex[color=white,opacity=0.8,x=2,y=0,label=3]{3} 
\Vertex[color=white,x=4,y=0,label=5]{5}
\Edge[style={->,> = latex'}](4)(1)
\Edge[style={->,> = latex'}](4)(2)
\Edge[style={->,> = latex'}](4)(6)
\Edge[style={->,> = latex'}](6)(3)
\Edge[style={->,> = latex'}](6)(5)
\end{tikzpicture}}
\end{center} \caption{The induced directed tree} \label{fig:tree}
\end{subfigure}
\caption{A stellar network and the induced graph of the information-domination~relation.\hfill$\triangle$}
\label{fig:2}
\end{figure}
\end{example}

The example above illustrates how layered information domination can emerge even in small hierarchical structures. Our next result shows that when such a structure appears as a \emph{stellar component} within a larger network, it may admit an extension that \textit{strictly} benefits the sender.

\begin{theorem}[Stellar]\label{thm:genstar}
Consider a network $g\in G(n)$ with a stellar component $C$, which has depth $\ell\ge 1$ and more than $k$ receivers. 
Suppose that there are at least $\ell$ receivers outside of $C$. Then there exists an extension of $g$ that strictly benefits the sender. Moreover, the sender's value for this extension is $V_k^n$.
\end{theorem}

\begin{remark}[Receivers' welfare]
    \label{rem:receivers}
    The construction in the proof of Theorem~\ref{thm:genstar} will show that the optimal experiment in the original network leads to outcome $x$ with probability $1$ conditional on the state being $X$. It follows from Lemma~\ref{lem:opt} in Appendix~\ref{ap:add} that the same holds in the denser network. Therefore, whenever the sender benefits from adding links to the network, the receivers are worse off in the following sense: the probability that the outcome is different from the true state increases. A similar phenomenon holds for \emph{all} results in the current section.
\end{remark}

We prove Theorem \ref{thm:genstar} in two steps. 
First, we show that when a sufficiently large stellar component is present, the sender cannot achieve the upper bound $V^n_k$ on the overall network. 
The reason is that persuading the central receiver in such a structure is often too costly: doing so requires revealing too much information in state $Y$, which in turn limits the sender’s ability to influence others.
Put differently, in an optimal experiment on a network with a sufficiently large stellar component, it is w.l.o.g. to always send the same message to the central node (i.e., the root of the induced tree). Thus, this node is persuaded only when unable to change the outcome. The sender can get a higher expected utility by effectively ignoring the central receiver and allocating effort among the remaining ones. Thus, the highest value the sender can obtain on such a network is strictly lower than $V^n_k$.

Second, we show that there is an extension of the network that benefits the sender. Figure \ref{fig:T1} illustrates the base for such an extension. It represents a network with 8 nodes, which comprises a (sufficiently big) stellar component (nodes  $\{1,2,\ldots,6\}$) and another component (nodes $7$ and $8$).
The idea is to extend the network by connecting receivers with different depths (i.e., receivers with different levels on the hierarchical order) to different receivers outside the component.  In Figure \ref{fig:T1}, receivers $1$, $2$, and $6$, who have depth $1$, are connected to receiver $7$; and receivers $3$ and $5$, who have depth $2$, are connected to receiver $8$. Note that we allow more than $\ell$ receivers to be present outside the stellar component (i.e., we can have more than the two receivers that appear in Figure \ref{fig:T1}), as long as the component remains sufficiently~large.

This construction breaks the information-dominating pairs within the stellar component. However, since we must also break the pairs involving nodes outside the stellar component $C$ (without creating others), and we do not assume anything about the links between the nodes outside $C$, the construction is sometimes much more intricate than depicted in Figure \ref{fig:T1}. It requires acknowledging different possibilities for the existing links between nodes outside $C$ and a separate treatment for $\ell=1$. Once all information-dominating pairs have been broken, we can invoke Lemma \ref{cor:info} to show that the value now is exactly $V^n_k$.

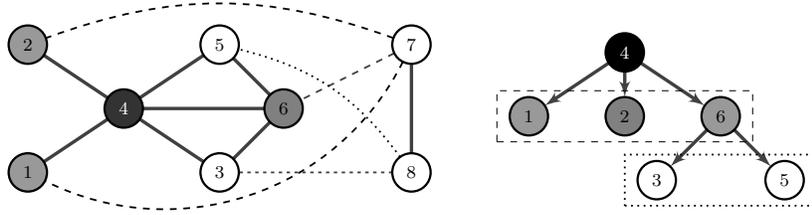
\begin{figure}[h]
\centering
\scalebox{.85}{
\begin{minipage}{0.45\textwidth}
\centering
\begin{tikzpicture}
\Vertex[color=gray,opacity=0.8,x=0,y=0,label=1]{1} 
\Vertex[color=gray,opacity=0.8,x=0,y=2,label=2]{2} 
\Vertex[color=white,x=3,y=2,label=5]{5} 
\Vertex[color=white,opacity=0.8,x=3,y=0,label=3]{3} 
\Vertex[color=black,opacity=0.8,x=1.5,y=1,label=\textcolor{white}{4}]{4}
\Vertex[color=gray,x=4,y=1,label=6]{6}
\Vertex[color=white,x=6,y=2,label=7]{7}
\Vertex[color=white,x=6,y=0,label=8]{8}
\Edge(1)(4)
\Edge(4)(5)
\Edge(2)(4)
\Edge(3)(4)
\Edge(4)(6)
\Edge(3)(6)
\Edge(8)(7)
\Edge(6)(5)
\Edge[style={dashed, thick}](7)(6)
\Edge[style={dotted, thick}](8)(3)
\path[-] (8) edge [bend right=20,style={dotted, thick}] node {} (5);
\path[-] (7) edge [bend right=20,style={dashed, thick}] node {} (2);
\path[-] (7) edge [bend left=45,style={dashed, thick}] node {} (1);
\end{tikzpicture}
\end{minipage}%
\begin{minipage}{0.45\textwidth}
\centering
\begin{tikzpicture}[rectup/.style={rectangle, draw,dashed, font=\Large,minimum height=8mm,minimum width=40mm},rectdown/.style={rectangle, dotted, thick, draw, font=\Large,minimum height=8mm,minimum width=30mm}]
\Vertex[color=black,x=1.5,y=2,label=\textcolor{white}{4}]{4} 
\Vertex[color=gray,x=1.5,y=1,label=2]{2} 
\Vertex[color=gray,opacity=0.8,x=0,y=1,label=1]{1} 
\Vertex[color=gray,opacity=0.8,x=3,y=1,label=6]{6}
\Vertex[color=white,opacity=0.8,x=2,y=0,label=3]{3} 
\Vertex[color=white,x=4,y=0,label=5]{5}
\Edge[style={->,> = latex'}](4)(1)
\Edge[style={->,> = latex'}](4)(2)
\Edge[style={->,> = latex'}](4)(6)
\Edge[style={->,> = latex'}](6)(3)
\Edge[style={->,> = latex'}](6)(5)
\node[rectup] (11) [xshift=1.5cm,yshift=1cm] { };
\node[rectdown] (22) [xshift=3cm,yshift=0cm] { };
\end{tikzpicture}
\end{minipage}}
\caption{The intuition behind the second step of the proof of Theorem \ref{thm:genstar}. The network consists of nodes $\{1,2,\dots,8\}$; nodes $\{1,2,\dots,6\}$ represent a stellar component. The left part illustrates the way the network (with solid lines) is extended (up to some required repairs; with dashed and dotted lines), while the right part shows the direct tree representing the information-domination relation. The color coding separates nodes with different depth. Lighter colors represent greater depth.}
\label{fig:T1}
\end{figure}

In organizational settings, the second step of the proof has a natural interpretation. After a merger, firms often make some employees \emph{boundary spanners} across multiple layers of the umbrella organization \citep*{noble2006role,williams2012we,colman2019postacquisition}; i.e., links are added to bridge the otherwise separated groups. Empirically, knowledge transfer across units arises when firms engineer direct communication, repeated visits, and broad introductions in post-acquisition integration \citep*{bresman1999knowledge}, which corresponds to adding links between receivers in the smaller component and receivers in different layers in the larger component. These cross-layer links are precisely the extensions that increase manipulability: they break information-domination relations and create shared exposure points the sender can target.

A second and increasingly common mechanism that resembles the construction in the proof of Theorem \ref{thm:genstar} is algorithmic link formation. ``People you may know'' or ``who to follow'' systems in social media propose connections that users would not have made otherwise, which leads to creating many new links \citep*{gupta2013wtf}. These systems implement link-prediction policies (e.g., supervised random walk, \citealp{backstrom2011supervised}) that may privilege some nodes over others. The democratic risk follows immediately if the platform also has an agenda. Users often (reasonably) believe that more connections lead to more information and thus better decisions, yet recommender policies can redirect \emph{who} becomes visible and \emph{which} messages circulate, to the sender’s advantage.

Recall that the star network is a special case of a stellar network with depth 1. Therefore, we immediately deduce the following.

\begin{corollary}[Star]\label{cor:star}
Let $g\in G(N)$ be a network containing a star component with more than $k$ and less than $n$ receivers. Then the value of $g$ is less than $V^n_k$, and there is an extension of $g$ that strictly benefits the sender and achieves the value of $V^n_k$.\footnote{As we shall see in Subsection~\ref{sub:non}, when the star component is the whole network, there is no beneficial extension for the sender.}
\end{corollary}

\noindent Since in a star network the sender does not attempt to persuade the central node (see Appendix \ref{app:star}), it is as if the sender is persuading $k$ out of a total of $n-1$ receivers.
Thus, $V^{n-1}_k$ is an upper bound of the value.\footnote{The conclusion for stars is robust to making the network directed in a more realistic way: if the center does not observe anyone but is observed by all peripheries (as with social media influencers), the sender still cannot achieve $V^n_k$; the underlying reason is unchanged -- the center can be made non-pivotal. As before, however, a suitable extension of the network restores the upper bound $V^n_k$.}
A similar result could be derived for networks with a component that has more than one center, say $m>1$: there are $m$ nodes in the component that are connected to each other and to all remaining nodes in the component, with no additional edges in the component. It follows from applying symmetry (Lemma~\ref{lem:symmetry}) to the centers that the sender can achieve at most $V^{n-m}_k$. This implies that for a fixed number of receivers $n$, the network value \emph{monotonically decreases} as $m$ increases, and becomes $V^p$ when $m$ is sufficiently large.\footnote{This is similar in spirit to a result of \cite{Candogan19} stating that the sender's payoff decreases when the degree of some nodes in the network increases.}

We can find other sequences of networks with monotonic decrease in the network value. One example is moving from the empty network to one with a large star component. However, Theorem \ref{thm:genstar} and Corollary~\ref{cor:star} show that this pattern might reverse, and at some point of extension the sender can strictly benefit from it. We now illustrate this non-monotonicity with an example involving a star component.

\begin{example}\label{ex:starlike}
Let $n=8$, $\lambda^0(X)=1/3$, and $k=4$.\footnote{Note that our assumption of $\lambda^0(X)<k/(n+k)$ in the paper is supposed to ensure that $V^n_k<1$, which implies that extending the network is \emph{strictly} beneficial to the sender when $V^n_k$ can be achieved via the extension. In this particular example, we allow for $\lambda^0(X)=1/3=k/(n+k)$, since we have $V^{n-1}_k<1$, and thus an extension achieving $V^n_k$ must be strictly beneficial to the sender.}
Consider the network $g$ in Figure \ref{fig:ex2} (without the dashed edge). By Corollary \ref{cor:star} and the discussion following it, for any $\pi\in\Pi$ it holds that $V^\pi_4(g)\leq V^{n-1}_4=V^7_4=11/12$.

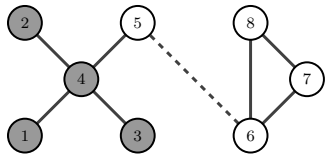
\begin{figure}[H]
\begin{subfigure}[t]{0.5\linewidth}
\begin{center}
\scalebox{.75}{
\begin{tikzpicture}
\Vertex[color=gray,opacity=0.8,x=0,y=0,label=1]{1} 
\Vertex[color=gray,opacity=0.8,x=0,y=2,label=2]{2} 
\Vertex[color=white,x=2,y=2,label=5]{5} 
\Vertex[color=gray,opacity=0.8,x=2,y=0,label=3]{3} 
\Vertex[color=gray,opacity=0.8,x=1,y=1,label=4]{4}
\Vertex[color=white,x=4,y=0,label=6]{6}
\Vertex[color=white,x=5,y=1,label=7]{7}
\Vertex[color=white,x=4,y=2,label=8]{8}
\Edge(1)(4)
\Edge(4)(5)
\Edge(2)(4)
\Edge(3)(4)
\Edge(7)(8)
\Edge(7)(6)
\Edge(8)(6)
\Edge[style={dashed}](6)(5)
\end{tikzpicture}}
\end{center}
\end{subfigure}
\begin{subfigure}[t]{0.5\linewidth}
\vspace{-2cm}
\begin{center}
\begin{tabular}{c|cc}
$ \pi' $ & $ X $ & $ Y $ \\
\hline
$\bar x$ & $1$ & $0$ \\[0.1cm]
$s$ & $0$ & $\frac{1}{2}$ \\[0.1cm]
$t$ & $0$ & $\frac{1}{2}$ \\
\end{tabular}
\end{center}
\end{subfigure}
\caption{Non-monotonicity of the network value in Example~\ref{ex:starlike}.}
\label{fig:ex2}
\end{figure}

\noindent Now consider the network $g'$ (with the dashed edge in Figure~\ref{fig:ex2}). Let $\bar x_i=x$ for all $i\in N$, $s\in S$ be s.t. $s_i=y$ for $i=6$ and $s_j=x$ for $j\in N\setminus\{6\}$.
Similarly, let $t\in S$ be s.t. $t_i=y$ for $i\in\{1,2,3\}$ and $t_j=x$ for $j\in N\setminus\{1,2,3\}$.
Consider $\pi'$ given in the right-hand side Figure~\ref{fig:ex2}. The outcome $x$ is implemented for any signal realization, and thus the optimal value strictly increases to $V^\pi_4(g')=V^8_4=1$.
\hfill$\triangle$
\end{example}

Example \ref{ex:starlike} highlights the importance of the network structure for the sender's gain from persuasion. 
It also illustrates an important difference between our results and \cite{Babichenko21}: while they show that $V^n_k$ can be achieved if there are \textit{no} information-dominating pairs, the sender can achieve $V^n_k$ on $g'$ \textit{even though} receiver $4$ information-dominates receivers $1$, $2$, and $3$. More broadly, while they provide a ranking of communication networks based on information domination that is valid for any prior and utility of the sender, we fix a threshold objective on a network with local spillovers and study sequences of extensions, identifying link additions that achieve $V^n_k$ (sometimes despite information domination).

We can interpret Example \ref{ex:starlike} from two different angles. 
For the sender, nodes with many sources of information (information hubs) are difficult to persuade. To increase the probability of persuasion (and similarly to break information-dominating pairs), the sender can either try to cut social ties, or, alternatively, to encourage more communication.\footnote{Cutting social ties to the extent of creating singletons makes the network approach an empty one and if there are sufficiently many singletons, the sender can achieve $V^n_k$ \citep{kerman2025bayesian}.} While the former strikes as a polarizing approach, the latter is usually perceived as unifying and democratic. However, both can be equally beneficial to the sender. From the receivers' perspective, forming more links seems like a natural improvement as it allows access to more sources of information (e.g., following more people on $\mathbb X$). Nevertheless, our result implies that receiving information from multiple sources might harm the receivers when these sources are highly correlated.\footnote{Increasing concentration of media ownership~\citep*{vizcarrondo2013measuring,noam2016owns} suggests that news from different sources might be highly correlated. Some theoretical studies in different contexts~\citep*{colla2010information,currarini2020strategic} show how correlation between information sources can have a negative effect on a decision maker.}

In Example~\ref{ex:starlike}, receiver $5$ forming a link with receiver $6$ allows the sender to employ a rougher partition of the network and reduce the influence of the information hub. While receiver $5$ has the good intentions to ``reach across the aisle'', work for ``bipartisan support'', try to ``build a bridge'' to opponents, etc., he achieves the opposite effect in terms of receivers' welfare. This example shows how the sender can leverage to her advantage natural properties of social networks, which often exhibit high degrees of clustering (e.g., around opinion leaders~\citep{Jackson07}).

We mentioned that the class of stellar networks is much more general than that of stars, as they permit communication among peripheries. We now introduce \emph{halo networks} that lie conceptually between the two (see Figure \ref{fig:halo} -- left). Like a stellar network, a halo has one central receiver connected to all others; however, the periphery forms a circle. The halo retains a simple hierarchical structure: the only information-domination relationship is between the center and the peripheries. This captures common real-world scenarios, such as moderated group chats (in which participants interact in a cycle while a central admin observes everyone) or work groups (in which team members collaborate directly, but still report to a single~manager).

\begin{definition}
A network $g\in G(n)$ is a \emph{halo} if it has exactly $2(n-1)$ links and there exists a sequence of distinct receivers $k_1, \dots , k_{n-1} \in N$ for whom $g_{k_1k_2} = g_{k_2k_3} =\dots= g_{k_{n-2}k_{n-1}} = g_{k_{n-1}k_1}= 1$, and the remaining receiver has degree $n-1$; see Figure~\ref{fig:halo}.
\end{definition}

\begin{figure}[t]
\begin{minipage}{0.45\textwidth}
\begin{center}
\scalebox{.75}{
\begin{tikzpicture}
\Vertex[color=white,opacity=0.8,x=0,y=0]{1}
\Vertex[color=white,opacity=0.8,x=0,y=2]{2}
\Vertex[color=white,opacity=0.8,x=3,y=2]{3}
\Vertex[color=white,opacity=0.8,x=3,y=0]{7}
\Vertex[color=white,opacity=0.8,x=1.5,y=1]{4}

\Edge(1)(4)
\Edge(2)(4)
\Edge(3)(4)
\Edge(7)(4)
\Edge(1)(2)
\Edge(2)(3)
\Edge(3)(7)
\Edge(7)(1)

\end{tikzpicture}}
\end{center}
\end{minipage}
\begin{minipage}{0.45\textwidth}
    \begin{center}
\scalebox{.5}{
\begin{tikzpicture}
\Vertex[color=white,opacity=0.8,x=1.5,y=0]{1}
\Vertex[color=white,opacity=0.8,x=1.5,y=2]{2}
\Vertex[color=white,opacity=0.8,x=1.5,y=-3]{3}
\Vertex[color=gray,opacity=0.8,x=3,y=-3]{4}
\Vertex[color=black,opacity=0.8,x=-1.5,y=1]{5}
\Vertex[color=black,opacity=0.8,x=4.5,y=1]{6}
\Vertex[color=gray,opacity=0.8,x=0,y=-3]{7}

\Edge(1)(5)
\Edge(2)(5)
\Edge(3)(5)
\Edge(4)(5)
\Edge(7)(5)
\Edge(1)(6)
\Edge(2)(6)
\Edge(3)(6)
\Edge(4)(6)
\Edge(5)(6)
\Edge(7)(6)
\Edge(3)(4)
\Edge(3)(7)

\end{tikzpicture}}
\end{center}
\end{minipage}
\caption{A \textit{halo} network (left) with $n=5$; a \textit{constellation} network (right) with two centers (the black nodes), three nodes with depth $1$ (white), and two nodes with depth $2$ (gray).}
\label{fig:halo}
\end{figure}
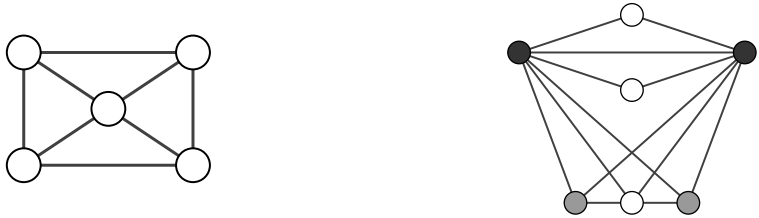

The construction in Theorem \ref{thm:genstar} proof can also be implemented for a network with a sufficiently large halo component; we provide the following corollary without proof. 

\begin{corollary}[Halo]
\label{cor:halo}
For any network $g\in G(N)$ containing a halo component with more than $k$ and less than $n$ receivers, there exists an extension that strictly benefits the sender and the extension has value $V^n_k$.
\end{corollary}

The networks that we have considered so far (apart from the discussion following Corollary~\ref{cor:star}) have a single center. However, various applications suggest also looking at networks in which information flows through multiple influential receivers, which motivates the notion of a \emph{constellation network} -- see Figure~\ref{fig:halo} (right).

\begin{definition} \label{def:constellation}
For $g\in G(N)$, $T\subsetneq N$, and $ i \in T$, let $g^i_T$ be the subnetwork of $g$ whose set of nodes is $D(g^i_T)=(N\setminus T)\cup\{i\}$, and for every $j,k \in D(g^i_T)$ it holds that $g_{jk}=g^i_{T\;jk}$. A \emph{constellation} $g$ is a network with a set of \emph{centers} $M\subsetneq N$ s.t. for all $i\in M$, it holds that $\delta_i=n-1$ and $g^i_M$ is a stellar network.
\end{definition}

Notice that a stellar network is a constellation with $\vert M\vert=1$. 
For $i\in M$, the \emph{depth of the subnetwork} $g^i$ of a constellation is the longest directed path in the tree representing information domination that it induces. 
The \emph{depth of a constellation} is the maximum depth of a subnetwork $g^i_M$, for $i\in M$. 

\begin{proposition}[Constellation]\label{pro:const}
Consider a network $g\in G(N)$ with a constellation component $C$ that has depth $\ell\geq 1$ and more than $k$ receivers. Let $M$ be the set of centers of $C$. Suppose that $\vert N\setminus C\vert\geq\max\{\ell,\vert M\vert\}+1$.
Then there exists an extension of the network that strictly benefits the sender. 
Moreover, the sender's value for this extension is $V^n_k$.
\end{proposition}

Since constellation networks are characterized by the presence of multiple centers, each of which information-dominates all the other receivers (including the other centers), the overall network does not form a directed tree (unlike stellar networks). Indeed, the information-domination relation among centers is mutual and creates cycles. However, for each center taken individually, the subnetwork obtained by removing the other centers forms a stellar network with an identical hierarchical structure rooted at that center. Proposition~\ref{pro:const} shows that such global domination by multiple receivers constrains the sender’s ability to persuade, preventing her from achieving the upper bound on the network value. However, extending the network by linking external receivers to the constellation in a way that breaks these mutual domination relations allows the sender to improve the value.

While constellations feature multiple dominating receivers within a single connected network (or component), a \emph{galaxy} takes this decentralization further by partitioning the network into fully disconnected components, each with its own locally dominating receiver. Unlike stellar and constellation structures, a galaxy does not require a hierarchical information flow within each component; it suffices that each component contains a receiver connected to all others in that component. This broader definition captures a wide range of fragmented communication settings in which the sender faces not just internal information constraints, but also structural barriers across subsets. It might represent segmented populations (geographic, social, or ideological) where individuals are influenced by local leaders, but are insulated from broader discourse. Galaxies, thus, serve as a useful benchmark for understanding how lack of inter-group communication can limit persuasion, and how adding targeted links across components can significantly enhance the sender’s gain from~persuasion.

\begin{definition} \label{def:galaxy}
    A \emph{galaxy} is a network consisting of (disjoint) components, each having a \emph{central node} that is connected to all other nodes at that component. The non-central nodes of a component are called \emph{peripheral nodes}.
\end{definition}

For example, a union of stellar components is a galaxy. Note that in a galaxy, there might be arbitrary links between peripheral nodes belonging to the same component.

\color{black}
\begin{theorem}[Galaxy]
    \label{thm:starlike}
    Suppose that $g\in G(N)$ is a galaxy with $\ell\ge 2$ components s.t.:
    \begin{inparaenum}[(i)]
    \item each component has a total number of nodes between $3$ and $\min\{k-1,n/2\}$;
    \item no subset of components has a union of size exactly $k$.
    \end{inparaenum}
    Then the sender's value is strictly less than $V^n_k$, and there exists an extension of $g$ in which the sender can achieve $V^n_k$.
\end{theorem}

A special case of a network satisfying the conditions of Theorem~\ref{thm:starlike} is a disjoint union of equal-sized star components, except for the case $k=n/2$, $\ell=2$, for which one can straightforwardly check that the sender's value for the network is $V_k^n$.

Since no receiver in one component can observe or be influenced by receivers in another one, the sender cannot coordinate behavior across components to reach the critical mass. In particular, if no union of components adds up to exactly $k$ receivers (which is a knife-edge case in realistic settings), then even selecting the ``right'' set of receivers to persuade becomes structurally infeasible. However, Theorem~\ref{thm:starlike} shows that by connecting these components in a ``smart'' way (effectively breaking the isolation), the sender can eliminate the information-domination relationships.

The proof is constructive by traversing the nodes in each component of $g$ when the components are sorted in a weakly increasing order of their size. When we traverse a node that has not been yet connected to any nodes outside of its own component in $g$ -- we connect it to the next available node (where availability means we can connect a node only to a single node outside of its component in $g$). Condition $(i)$ ensures that the process might get stuck only while we traverse the last component, and even if it does, the network can be modified to remove all the information-dominating pairs.

Finally, we consider even more constrained networks for the sender, \emph{cluster networks}, in which each component is a cluster (a complete subnetwork). Like galaxies, cluster networks consist of disjoint subsets, but now each set exhibits the densest internal communication possible -- every receiver within a component observes the messages of all others. This structure is relevant for various applications. In political contexts, it resembles ideological clusters or partisan echo chambers in which individuals deliberate primarily within their in-group, reinforcing common beliefs. In consumer markets, cluster networks can arise when close groups strongly influence each other’s purchases; this makes persuasion especially difficult. Receivers are likely to interpret signals similarly and coordinate their actions, limiting the sender’s ability to exploit belief asymmetries. However, we show that adding cross-group links (e.g., through ads or referrals) can still strictly improve the sender’s value.

\color{black}
\begin{proposition}[Clusters]
    \label{pro:cliques}
    Consider a network $g$ consisting of $q\ge 2$ disjoint clusters of size $p\ge 2$ each, and assume $n/2<k<n$. Then there exists an extension of the network that strictly benefits the sender. Moreover, the sender's value for this extension is $V_k^n$.
\end{proposition}

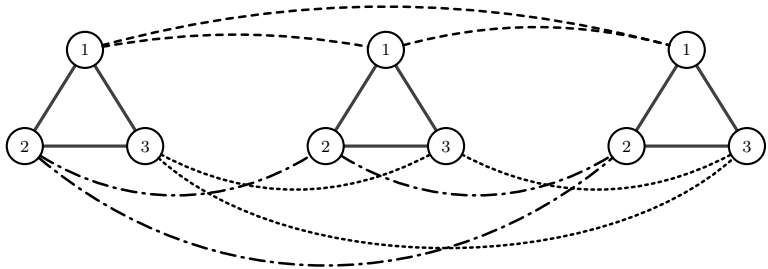
\begin{figure}[h]
\centering
\scalebox{.8}{
\begin{tikzpicture}

\Vertex[color=white,opacity=0.8,x=0.0,y=2.6,label=1]{L1}
\Vertex[color=white,opacity=0.8,x=-1.0,y=1.0,label=2]{L2}
\Vertex[color=white,opacity=0.8,x= 1.0,y=1.0,label=3]{L3}
\Edge(L1)(L2)\Edge(L1)(L3)\Edge(L2)(L3)

\Vertex[color=white,opacity=0.8,x=5.0,y=2.6,label=1]{M1}
\Vertex[color=white,opacity=0.8,x=4.0,y=1.0,label=2]{M2}
\Vertex[color=white,opacity=0.8,x=6.0,y=1.0,label=3]{M3}
\Edge(M1)(M2)\Edge(M1)(M3)\Edge(M2)(M3)

\Vertex[color=white,opacity=0.8,x=10.0,y=2.6,label=1]{R1}
\Vertex[color=white,opacity=0.8,x= 9.0,y=1.0,label=2]{R2}
\Vertex[color=white,opacity=0.8,x=11.0,y=1.0,label=3]{R3}
\Edge(R1)(R2)\Edge(R1)(R3)\Edge(R2)(R3)


\draw[dashed, line width=1.2pt, line cap=round]
  (L1) to[out=10, in=170, looseness=0.9] (M1);
\draw[dashed, line width=1.2pt, line cap=round]
  (M1) to[out=15, in=165, looseness=0.9] (R1);
\draw[dashed, line width=1.2pt, line cap=round]
  (L1) to[out=15, in=165, looseness=0.9] (R1);

\draw[dash pattern=on 7pt off 3pt on 1pt off 3pt, line width=1.2pt, line cap=round]
  (L2) to[out=-32, in=-148, looseness=0.95] (M2);
\draw[dash pattern=on 7pt off 3pt on 1pt off 3pt, line width=1.2pt, line cap=round]
  (M2) to[out=-32, in=-148, looseness=0.95] (R2);
\draw[dash pattern=on 7pt off 3pt on 1pt off 3pt, line width=1.2pt, line cap=round]
  (L2) to[out=-40, in=-140, looseness=1] (R2);

\draw[dotted, line width=1.2pt, line cap=round]
  (L3) to[out=-28, in=-152, looseness=0.95] (M3);
\draw[dotted, line width=1.2pt, line cap=round]
  (M3) to[out=-28, in=-152, looseness=0.95] (R3);
\draw[dotted, line width=1.2pt, line cap=round]
  (L3) to[out=-42, in=-138, looseness=0.8] (R3);

\end{tikzpicture}}
\vspace{-0.5cm}
\caption{The construction in the proof of Proposition \ref{pro:cliques} for a cluster network with three components, each having 3 receivers. The extension forms rings across clusters between receivers with the same index.}
\label{fig:clusters}
\end{figure}

First, note that Proposition \ref{pro:cliques} formalizes the intuition that segmentation in the society limits the persuasion power of the sender, as $V^n_k$ cannot be achieved. Second, and perhaps more interestingly, it shows that forming bridges between clusters in a particular way can overturn the limitation.
In the proof, the construction of the extension can be thought of as connecting \emph{role-equivalent} receivers across clusters (see Figure \ref{fig:clusters}), e.g., admins of different groups or managers across departments, which is also a feature in social media platforms that nudge similar users to connect. Technically, these bridges break information domination and reintroduce enough independence across individuals that the sender can design an experiment that achieves the quota with the optimal value. The proof makes this precise by adding rings that connect same-index members across cliques and showing that the resulting network has no information-dominating pairs; thus, $V^n_k$ is attainable.\footnote{Note that the construction could have been done by forming rings between receivers with different indexes as well; however, the rings would still need to consist of disjoint sets of receivers.} Hence, segmentation (dense within-group talk, sparse cross-group links) initially shields against finely targeted persuasion by coordinating beliefs and actions at the clique level.
However, selective bridging by organizational or platform mechanisms can increase manipulability by allowing the sender to employ private communication more effectively. That is, more communication between groups can make the population more manipulable.

\subsection{Networks with No Beneficial Extensions}
\label{sub:non}
To complement the analysis, we examine some cases in which extending the network \emph{cannot} benefit the sender. Specifically, we highlight configurations in which structural features of the network leave the sender with no room for strategic improvement, regardless of how the network is extended. It is obvious that starting from the empty network and adding any link to it will result in no benefit to the sender (and no harm to the receivers). In fact, in any network in which the value $V^n_k$ is achieved, adding a link will not be beneficial to the sender. When the sender cannot achieve $V^n_k$, it is more subtle to determine whether the network has a beneficial extension. In the following examples, we highlight cases in which, although the sender's value is below $V^n_k$, no extension benefits her. The examples share some structural features that limit the sender's flexibility: they involve information-domination relations that are highly concentrated and difficult to disrupt through addition of links. The sender’s disadvantage arises from the presence of clusters or dominant subsets that collectively suppress belief asymmetries. Examples \ref{ex:no_detrimental} and \ref{ex:no_detrimental2} illustrate these features in two distinct ways. Finding natural sufficient conditions that are not too restrictive s.t. the sender cannot benefit from extending the network remains an interesting open question.

\begin{example}\label{ex:no_detrimental}
Consider a network $g$ with a cluster $C$ of size $1\le \ell\le n-1$ ($n\ge 2$) s.t. each of the remaining $n-\ell$ nodes is connected to \emph{all} the nodes in $C$; see Figure~\ref{fig:ex3} (left).\footnote{A simple example of such $g$ with $\ell=1$ is a star.} Take any $k\geq \lfloor \frac{n+1}{2}\rfloor$. It follows from Lemma~\ref{lem:symmetry} that the problem is equivalent to communicating with $n-\ell+1$ receivers who are in a star network. By Proposition~\ref{pro:star2} in Appendix~\ref{ap:add}, the sender's value is bounded from above by $V_k^{n-\ell}$. 
Moreover, such utility is trivially achieved when there are no links between receivers outside of $C$, since the sender can treat $g\setminus C$ as the empty network. Hence, no extension of $g$ benefits the sender.
\end{example}

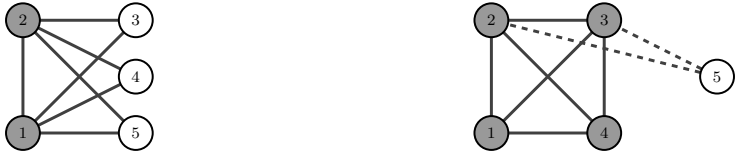
\begin{figure}[h]
\begin{minipage}{0.45\textwidth}
\begin{center}
\scalebox{.75}{
\begin{tikzpicture}
\Vertex[color=gray,opacity=0.8,x=0,y=0,label=1]{1} 
\Vertex[color=gray,opacity=0.8,x=0,y=2,label=2]{2} 
\Vertex[color=white, x=2,y=2,label=3]{3} 
\Vertex[color=white,x=2,y=1,label=4]{4} 
\Vertex[color=white,x=2,y=0,label=5]{5}
\Edge(1)(5)
\Edge(1)(2)
\Edge(1)(3)
\Edge(1)(4)
\Edge(2)(5)
\Edge(2)(3)
\Edge(2)(4)
\end{tikzpicture}}
\end{center}
\end{minipage}
\begin{minipage}{0.45\textwidth}
    \begin{center}
\scalebox{.75}{
\begin{tikzpicture}
\Vertex[color=gray,opacity=0.8,x=0,y=0,label=1]{1} 
\Vertex[color=gray,opacity=0.8,x=0,y=2,label=2]{2} 
\Vertex[color=gray,opacity=0.8,x=2,y=2,label=3]{3} 
\Vertex[color=white,x=4,y=1,label=5]{4} 
\Vertex[color=gray,opacity=0.8,x=2,y=0,label=4]{5}
\Edge(1)(5)
\Edge(1)(2)
\Edge(1)(3)
\Edge(2)(5)
\Edge(2)(3)
\Edge(5)(3)
\Edge[style={dashed}](4)(3)
\Edge[style={dashed}](4)(2)
\end{tikzpicture}}
\end{center}
\end{minipage}
\caption{An example of a constellation described in Example \ref{ex:no_detrimental}, where $\ell=2$ and $C$ consists of the gray receivers (left); an example of the network from Example \ref{ex:no_detrimental2}, where $C$ consists of the gray receivers (right).}
\label{fig:ex3}
\end{figure}

In Example~\ref{ex:no_detrimental}, the information domination originates in a set $C$ of receivers that (weakly) dominate each other and dominate all the other receivers. The key point is that adding any subset of links will not break any information-dominating pair. It happens since one can only add links between nodes that: \begin{inparaenum}[(i)]
    \item observe identical sets of channels in the original network; and
    \item do not information-dominate any nodes in the original network.
\end{inparaenum}

\begin{example}\label{ex:no_detrimental2}
Consider a network $g$ consisting of a cluster $C$ of size $n-1=2k-2\ge 2$ and an isolated receiver $i$; see Figure~\ref{fig:ex3} (right). Since the receivers in $C$ have the same posterior and $|C|\ge k$, optimal persuasion is equivalent to public persuasion of the $n-1$ receivers in $C$ while ignoring $i$. Therefore, the sender's value is $V^p=2\lambda^0(X)$. Adding links between $i$ and receivers in $C$ partitions the receivers in $C$ into two sets $C_1,C_2$ in terms of the messages that they observe: the receivers in $C_1$ observe exactly the messages of all the receivers in $C$, and the receivers in $C_2$ observe all messages.

If $|C_2|\ge k$, then the outcome is $x$ iff the members of $C_2$ take action $x$. 
If $|C_1|\ge k$, then the outcome is $x$ iff the members of $C_1$ take action $x$. 
In both cases, the sender's optimal expected utility is again $V^p=2\lambda^0(X)$. 
Finally, we might have $|C_1|=|C_2|=k-1$. 
This case is equivalent from the sender's perspective to persuasion with $n=3$, $k=2$ with one receiver (representing $C_2$) connected to the two other receivers (representing $C_1$ and $i$) -- i.e., the network is a star with a central node and two peripheral nodes. By Proposition~\ref{pro:star2}, the sender w.l.o.g. can ignore the central node (representing $C_2$), which yields an expected utility of $V_2^2=2\lambda^0(X)$.
\end{example}

Example~\ref{ex:no_detrimental2} is different in spirit from Example~\ref{ex:no_detrimental}. The proof essentially reduces the network to the case $n=3,k=2$ with the original network containing a single edge. In this case, it is possible to break an existing information-dominating pair by adding a link -- and still, it is not beneficial for the sender. The key intuitions behind Example~\ref{ex:no_detrimental2} are: (i) already in the original network, the information-domination relation imposes very strict limitations on the sender (in fact, removing the isolated node would result in a network corresponding to public signaling); and (ii) there is a lot of symmetry in the network. Due to the combination of these factors, in the original network there are only two types of receivers in terms of the information they observe; adding more links can only add at most one more type. Thus, the sender's ability to differentiate between the receivers is very limited, and adding more links does not help.

\section{Conclusion}\label{sec:con}
We study a Bayesian persuasion model with multiple binary-action receivers embedded in a network s.t. each receiver observes not only his own private message, but also those of his neighbors. The sender aims to persuade a critical mass of receivers to take a particular action. While conventional wisdom suggests that more communication among receivers should help them resist manipulation, we show that this is not necessarily the case. In fact, the sender can strictly benefit from increased communication among the receivers, even when it appears to further constrain her. 

Our analysis shows that the sender’s value does not change monotonically with the network density. We characterize a variety of network structures having relevant real-world applications for which adding links can unexpectedly improve the sender’s ability to persuade. These gains arise because local communication can either limit variation in receiver's beliefs or create opportunities for the sender to differentiate messages strategically, depending on the network’s structure. Moreover, we identify some networks in which no extension improves the sender’s value.

These findings have practical implications for settings in which influence and information design intersect with social structure. For instance, in marketing, firms may find that semi-connected consumer sets are more susceptible to persuasion than isolated or fully-connected ones. In political campaigns, the presence of localized communication (e.g., between voters of the same party) can enable politicians to more effectively tailor persuasive strategies. Similarly, platform designers or regulators should be aware that promoting more communication is not always protective for users: it can inadvertently increase their vulnerability to~manipulation.

\appendixtitleon
\appendixtitletocon
\begin{appendices}
\section{Additional Results}
\label{ap:add}
\subsection{Experiments Achieving the Upper Bound of the Value}
We state a technical lemma that summarizes the properties of optimal experiments that achieve the upper bound of the value, $V^n_k$. It turns out that conditional on the state being $X$, the outcome is $x$ with probability $1$. In addition, conditional on the state being $Y$, either exactly $k$ receivers have the uniform posterior and $n-k$ receivers have the zero posterior belief that the state is $X$, or all the receivers have the zero posterior belief that the state is $X$.

\begin{lemma}\label{lem:opt}
Consider an experiment achieving the sender's optimal value of $V_k^n=(n/k+1)\lambda^0(X)<1$ and let $\widetilde V=(n\lambda^0(X))/(k\lambda^0(Y))$. 
Then conditional on the true state being $X$, the outcome is $x$ with probability $1$. Moreover, conditional on the true state being $Y$, it holds that:
\begin{enumerate}[(i)]
    \item the sender's expected utility is $\widetilde V$;
    \item with probability $\widetilde V$, there are $k$ receivers with posterior belief 1/2 and $n-k$ receivers with posterior belief 0 that the state is $X$;
    \item with probability $1-\tilde V$, all receivers have posterior belief 0 that the state is $X$.
\end{enumerate}
\end{lemma}

\begin{proof}
We know from~\cite{kerman2024persuading} that $V_k^n=(n/k+1)\lambda^0(X)$. Note that conditional on the state being $X$, the sender can get expected utility of at most $1$, which happens if whenever the state is $X$, at least $k$ receivers have posteriors $\lambda^{s,g}_i(X)\ge 1/2$. We deduce from the law of total expectation that the sender's expected utility conditional on the state being $Y$ is at most $\frac{V_k^n-1\cdot\lambda^0(X)}{1-\lambda^0(X)}=\frac{n\lambda^0(X)}{k\left(1-\lambda^0(X)\right)}$. Since the lemma statement describes how to achieve such utility, the first bullet follows.

To prove the second and third bullets, replace the given network with the empty network (imitating the signaling from the original network by sending private messages with all the required information to each receiver) and perform the following transformation on an experiment that achieves $V_k^n$: whenever the state is $X$ and a certain receiver gets a signal leading to a posterior probability for $X$ strictly smaller than $1/2$, reveal the true state to that receiver instead. After the transformation, all the receivers have posteriors $\lambda^{s,g}_i(X)\ge 1/2$ when the state is $X$. The transformation cannot affect uniform posteriors, but it can induce posteriors $\lambda^{s,g}_i(X)=0$. Therefore, it is not immediately clear that this transformation is w.l.o.g. However, if this transformation affects the experiment, then for some receiver it must induce with a positive probability a posterior $\lambda^{s,g}_i(X)=1$ when the state is $X$. We shall show that no such posteriors are induced. Note that the expectation of the posterior of each specific receiver is $\lambda^0$ by~\cite{kamenica2011bayesian}. Therefore, the expectation of the sum of all the receivers' posterior probabilities for $X$ is $n\lambda^0(X)$. We deduce from the first bullet that in any sender's optimal experiment in the empty network, the probability that at least $k$ receivers have posterior $\lambda^{s,g}_i(X)\ge 1/2$ conditional on the state being $Y$ is $\frac{n\lambda^0(X)}{k\left(1-\lambda^0(X)\right)}$. Thus, the (unconditional) expectation of the sum of all the receivers' posterior probabilities for $X$ is at least:

\begin{equation*}
    \frac{n\lambda^0(X)}{k\left(1-\lambda^0(X)\right)} \cdot \frac{1}{2}\cdot k \cdot \left(1-\lambda^0(X)\right) + \frac{1}{2}\cdot n \cdot \lambda^0(X) = n\lambda^0(X).
\end{equation*}

Hence, equality must hold, which implies that our process performed no transformations on the experiment at all (as there are no posteriors with $\lambda^{s,g}_i(X)=1$). This immediately implies the second and the third bullets.
\end{proof}

\subsection{Optimality in Networks with a Star Component}\label{app:star}
We prove that ignoring the central node in a sufficiently large star component cannot hurt the sender.

\begin{proposition}
    \label{pro:star2}
    Let $g\in G(N)$ have a star component with set of receivers $C$ and $\vert C\vert>k$. Let $g'=g\setminus \{c\}$, where $c$ is the central node of $C$. Then for any $\pi\in\Pi$, it holds that $V^\pi_k(g)=V^\pi_k(g')$.
\end{proposition}
\begin{proof}

\noindent We first define an \emph{anchor}, which will be used in the rest of the proof. 
\begin{definition} \label{def:acnhor}
Let $\pi\in\Pi$. 
A signal $s\in S^{\pi}$ is an \emph{anchor} if $\pi(s|X)\lambda^0(X)\geq\pi(s|Y)\lambda^0(Y)$. 
The set of all anchors is denoted by $An(\pi)$.
\end{definition}

Denote the center of the star component by $c\in N$.
We first show that transferring the information of $c$ to the periphery nodes leads to no loss of information to $c$, as he can uniquely reconstruct the information from the combination of messages of the peripheral nodes. 
This means that $c$ can be sent a single message in all signals without changing the value. 
\begin{lemma}\label{lem:center}
Let $g\in G(N)$ have a star component with set of receivers $C$ and $\vert C\vert>k$. 
For any $\hat \pi \in \Pi$ and $\hat s\in S^{\hat \pi}$ s.t. $\alpha^{\hat \pi,g}_c(\hat s)=x$, there exists $\pi \in S^\pi$ s.t. $|S^\pi_c|=1$ and $V^{\hat \pi}_k(g)= V^{\pi}_k(g)$.
\end{lemma}

\begin{proof} Note that for two anchors $s,t\in An(\hat\pi)$ with $s_c\neq t_c$, it holds that $A^{\hat \pi}_c(g,s)\cap A^{\hat \pi}_c(g,t)=\emptyset$. 
Let $S'\subseteq S$. 
Define a bijection $\tau:S^{\hat\pi}\rightarrow S'$ s.t. $\tau(s)=s'$ if 
\begin{inparaenum}[(i)]
\item $s'_c=x$, 
\item for $j\in C\setminus \{c\}$ it holds that $s'_j=(s_j,s_c)$, and 
\item for $\ell\in N\setminus C$ it holds that $s'_{\ell}=s_{\ell}$. 
\end{inparaenum}
That is, in signals in $S'$ the center always observes $x$ and the messages of nodes $C\setminus \{c\}$ are modified so that they contain the information previously provided by $c$ in signal $s$. 
So, the information that $c$ reveals to nodes in $C\setminus \{c\}$ is shifted to them while $c$ observes the same message in every signal. 

For any $\omega\in\Omega$ and $s'\in S'$ s.t. $\tau(s)=s'$, let $\pi\in\Pi$ be defined by $\pi(s'\vert\omega)=\hat\pi(\tau^{-1}(s')\vert\omega)$.
As the probabilities of corresponding signals are the same under $\pi$ as under $\hat\pi$ and $c$'s information under $\hat\pi$ is shifted to nodes in $C\setminus \{c\}$ under $\pi$ (which are observed by $c$), $c$'s action does not change. 
Moreover, the actions of nodes in $C\setminus \{c\}$ and in $N\setminus C$ do not change either. 
To see this, note that for any $i\in N$ and $t'\in A^{\pi}_i(g,s')$ there exists $t\in A^{\hat\pi}_i(g,s)$ s.t. $\tau(t)=t'$. 
Together with the definition of $\tau$, it implies that $\sum_{t'\in A^{\pi}_i(g,s')}\pi(t'\vert\omega)=\sum_{t\in A^{\hat\pi}_i(g,s)}\hat\pi(t\vert\omega)$. 
Thus, every node has the same posterior upon observing $s\in S^{\hat\pi}$ and $\tau(s)\in S^{\pi}$. Hence, $V^{\pi}_k(g)=V ^{\hat \pi}_k(g)$. 
\end{proof}

Next, we show that whenever there is an experiment $\hat \pi$ with a signal in which $c$ chooses $x$, there is another experiment which \emph{preserves the value} and in which in all signals where $c$ chooses $x$, all peripheral receivers also choose $x$.
\begin{lemma}\label{lem:center2}
Let $g\in G(N)$ have a star component with set of receivers $C$ and $\vert C\vert>k$.
For any $\hat \pi \in \Pi$ and $\hat s\in S^{\hat \pi}$ s.t. $\alpha^{\hat \pi,g}_c(\hat s)=x$, there exists $\pi \in S^\pi$ where for every $s\in S^\pi$ s.t. $\alpha^{\pi,g}_c(s)=x$ it holds that $\alpha^{\pi,g}_i(s)=x$ for all $i\in C$ and $V^{\hat \pi}_k(g)= V^{\pi}_k(g)$.
\end{lemma}
\begin{proof} Put simply, by means of unique messages, every peripheral node in $C$ can uniquely identify the signals in which $c$, which has \emph{strictly} more information than a peripheral node, chooses $x$. 
Hence, all peripheral nodes are persuaded whenever $c$ is.

Suppose that there is a signal $t\in S^{\hat \pi}$ in which $c$ chooses $x$. 
Hence, there is at least one anchor $s \in An(\hat \pi)$ with $s_i=t_i$ for all $i \in C$ and for every $r\in S^{\hat \pi}$ s.t. $r_i=s_i$ for all $i \in C$, $c$ also chooses $x$.
The action patterns of nodes in $C$ in such signals are:
\begin{inparaenum}[($a$)]
\item all $x$;
\item $c$ chooses $x$ and zero or more nodes in $C\setminus \{c\}$ choose $x$.
\end{inparaenum}

Consider case $(b)$. 
It must be true that if some node $\ell\in C\setminus \{c\}$ chooses $y$ this is because it associates $t$ with \emph{more} signals than $c$.
In other words, in all signals $r\in S^{\hat \pi}$ where $(r_{\ell},r_c)=(t_{\ell},t_c)$ node $\ell$ chooses $y$ and this \emph{includes} the signals in which $c$ \emph{does not} choose $x$. 
As this also includes the associated anchors, $A^{\hat\pi}_c(g,t)\subsetneq A^{\hat\pi}_{\ell}(g,t)$. 

Trivially, for every $s,t\in S^{\hat \pi}$ with $s_c\neq t_c$ and $i\in C$, it holds that $A^{\hat\pi}_i(g,s)\cap A^{\hat\pi}_i(g,t)=\emptyset$. Thus, if $c$ receives different messages in different signals, these signals belong to \emph{disjoint} association sets, and this observation holds for every $i\in C$.

Let $T$ be the set of signals in which receiver $\ell$ chooses $y$ and $c$ chooses $x$. 
That is, 
\begin{equation*}
T=\left\{t\in S^{\hat\pi}\vert\: \alpha^{\hat\pi,g}_{\ell}(t_{\ell}(g))=y \text{ and } \alpha^{\hat\pi,g}_c(t_c(g))=x\right\}.
\end{equation*}

\noindent Define a bijection s.t. for signals in $\hat\pi$ in which $\ell$ chooses $y$ and $c$ chooses $x$, the message of $\ell$ is changed to a \emph{unique} message per each information set of $c$ and keep all other messages the same. Formally, let $T'\subsetneq S$ and define $\phi:T\rightarrow T'$ s.t. for any $t\in T$ it holds that $\phi(t)=t'$ if $t'_{\ell}=t_c(g)\in S'_{\ell}\setminus S^{\hat\pi}_{\ell}$ and $t'_{-{\ell}}=t_{-{\ell}}$.
Now for any $\omega\in\Omega$ define a new experiment $ \pi\in\Pi$, which transforms the signals in $T$ according to $\phi$ and keeps other signals the same while preserving the probability weights:
\vspace{-0.25cm}
\begin{align*}
\pi\left(s'\vert\omega\right) = \begin{cases}
\hat\pi(s'\vert\omega) & \text{if } s'\in S^{\hat\pi}\setminus T,\\
\hat\pi(\phi^{-1}(s')\vert\omega) & \text{if } s'\in T.
\end{cases}
\vspace{-0.25cm}
\end{align*}

\noindent Let $s'\in S^{\pi}$ be s.t. $\phi(s)=s'$ for some $s\in T$. 
Then, 
\vspace{-0.25cm}
\begin{align*}
&\lambda^{s',g}_{\ell}(X)=\frac{\sum_{t'\in A^{\pi}_{\ell}(g,s')}\pi(t'\vert X)\lambda^0}{\sum_{\omega\in\Omega}\sum_{t'\in A^{\pi}_{\ell}(g,s')}\pi(t'\vert\omega)\lambda^0(\omega)}=\frac{\sum_{t'\in A^{\pi}_{\ell}(g,s')}\hat\pi(\phi^{-1}(t')\vert X)\lambda^0}{\sum_{\omega\in\Omega}\sum_{t'\in A^{\pi}_{\ell}(g,s')}\hat\pi(\phi^{-1}(t')\vert\omega)\lambda^0(\omega)}\\
&=\frac{\sum_{t\in A^{\hat\pi}_{\ell}(g,s)\cap A^{\hat\pi}_{c}(g,s)}\hat\pi(t\vert X)\lambda^0}{\sum_{\omega\in\Omega}\sum_{t\in A^{\hat\pi}_{\ell}(g,s)\cap A^{\hat\pi}_{c}(g,s)}\hat\pi(t\vert\omega)\lambda^0(\omega)}=\frac{\sum_{t\in A^{\hat\pi}_{c}(g,s)\subseteq T}\hat\pi(t\vert X)\lambda^0}{\sum_{\omega\in\Omega}\sum_{t\in A^{\hat\pi}_{c}(g,s) \subseteq T}\hat\pi(t\vert\omega)\lambda^0(\omega)}\geq\frac12,
\end{align*}

\noindent where $\phi(t)=t'$ and the third equality follows from the definition of $\phi$; $A^{\hat\pi}_{\ell}(g,s)\cap A^{\hat\pi}_c(g,s)=A^{\hat\pi}_c(g,s) \subseteq T$ follows from $A^{\hat\pi}_c(g,t)\subsetneq A^{\hat\pi}_{\ell}(g,t)$ and the inequality follows from the definition of case $(b)$.
Similarly, it holds that $\lambda^{s',g}_c(X)\geq 1/2$. 
This implies that in $\pi$ node $\ell$ will choose $x$ whenever $c$ chooses $x$ in $\pi$.
Additionally, node $c$ will still choose $x$ in the corresponding signals in $\hat\pi$ and $\pi$. 
Thus, the transformation does not change the action of $c$ in any signal, it only \emph{increases} the number of $x$ actions. 
Observe that for $s\in A^{\pi}_{\ell}(g,t)\setminus A^{\pi}_c(g,t)$ s.t. $t\in T$, it holds that $\alpha^{\hat\pi,g}_{\ell}(t_{\ell}(g))=y$ and the transformation will not decrease the value, as in such $s$ nodes $\ell$ and $c$ must already be choosing $y$. 
Hence, $V^{\pi}_k(g)=V ^{\hat \pi}_k(g)$. 
\end{proof}

\noindent By Lemma \ref{lem:center}, the center of a star component can get the same message in all signals without affecting the value. Moreover, by Lemma \ref{lem:center2}, whenever $c$ chooses $x$, \emph{all peripheral nodes} choose $x$. Hence, $c$'s action will \emph{never be decisive}, since $|N\setminus C|<k$. Even persuading all $N\setminus C$ will require persuading a member of $C$, which cannot be $c$ -- otherwise, there would be \emph{strictly more} receivers than the critical mass who choose $x$. The central node thus becomes a \emph{dummy player} who should never be persuaded.
\end{proof}

\section{Proofs} \label{ap:proofs}

\begin{proof}[Proof of Proposition \ref{pro:emptynetwork}]
We start by proving the upper bound. Let $\pi\in\Pi$.
For each $i\in N$, denote $c(i)=\vert S^{\pi}_i(g)\vert$. 
Let $R(i)=\left\{m^1_i,\ldots,m^{c(i)}_i\right\}\subseteq S_i$ be a set of distinct messages for $i$. 
For any $i\ne j$ in $N$, $q\in\left\{1,\ldots,c(i)\right\}$, and $q'\in\left\{1,\ldots,c(j)\right\}$, assume w.l.o.g.~that $m^q_i\neq m^{q'}_j$.

For each $i\in N$, let $\phi_i:S^{\pi}_i(g)\rightarrow R(i)$ be a bijection -- i.e., each \emph{information set} of $i$ is mapped to a unique \emph{message} in $R(i)$.
For each $\omega\in\Omega$ and $s'\in S$, define $\pi'\in\Pi$: 
\begin{align*}
\pi'\left(s'\vert\omega\right) = \begin{cases}
\pi(s\vert\omega) & \text{if } \phi_i(s_i(g))=s'_i,\quad\forall i\in N,\\
0 & \text{otherwise}.
\end{cases}
\vspace{-0.75cm}
\end{align*}
\noindent The definition of $\pi'$ implies that there is a bijection $\phi:S^{\pi}\rightarrow S^{\pi'}$ s.t. for each $i\in N$, $\phi(s)=s'$ iff $\phi_i(s_i(g))=s'_i$. 
Hence, $\pi'$ is an experiment under the empty network $g_0$. We want to show that the value of $\pi'$ under the empty network is equal to the value of $\pi$ under $g$ -- i.e., $V^{\pi'}_k(g_0)=V^{\pi}_k(g)$. 
It suffices to show that each receiver $i$ has the same posterior upon observing $s_i(g)$ under $\pi$ and observing $\phi_i(s_i(g))$ under $\pi'$. 
Let $s'\in S^{\pi'}$ be s.t. $s'_i\in\left\{m^1_i,\ldots,m^{c(i)}_i\right\}$. 
For any $\omega\in\Omega$:

\begin{align*}
\lambda^{s',g_0}_i(\omega)&=\frac{\sum_{s\in S^{\pi'}:s_i=s'_i}\pi'(s\vert\omega)\lambda^0(\omega)}{\sum_{\omega'\in\Omega}\sum_{s\in S^{\pi'}:s_i=s'_i}\pi'(s\vert\omega')\lambda^0(\omega')}=\frac{\sum_{s\in S^{\pi}:s_i(g)=\phi^{-1}\left(s'_i\right)}\pi(s\vert\omega)\lambda^0(\omega)}{\sum_{\omega'\in\Omega}\sum_{s\in S^{\pi}:s_i(g)=\phi^{-1}\left(s'_i\right)}\pi'(s\vert\omega')\lambda^0(\omega')}\\
&=\frac{\sum_{s\in A^{\pi}_i\left(g,\phi^{-1}(s')\right)}\pi(s\vert\omega)\lambda^0(\omega)}{\sum_{\omega'\in\Omega}\sum_{s\in A^{\pi}_i\left(g,\phi^{-1}(s')\right)}\pi(s\vert\omega')\lambda^0(\omega')}=\lambda^{\phi^{-1}(s'),g}_i(\omega). 
\end{align*}

\noindent Thus, for each $s\in S^{\pi}$, it holds that $\alpha^{\pi,g}(s)=\alpha^{\pi',g_0}(\phi(s))$. 
Hence, $V^{\pi'}_k(g_0)=V^{\pi}_k(g)$. 
Since any $\pi\in\Pi$ on some $g$ can be replicated on the empty network, $V^n_k\geq V^{\pi}_k(g)$.

To prove the lower bound on $V^{\pi}_k(g)$ for an optimal experiment $\pi$ on $g$, note that in any network, it is possible to imitate public signaling by transmitting the same message to all the receivers. Therefore, $V^\pi_k(g)\geq V^p$ for any optimal $\pi\in\Pi$.
\end{proof}
\begin{proof}[Proof of Lemma \ref{lem:symmetry}]
Let $\pi\in\Pi$ and let $g\in G(N)$ and $i,j\in N$ be s.t. $ N_i(g)= N_j(g)$. We shall construct $\pi'\in\Pi$ s.t. for any $s\in S^{\pi'}$, it holds that $s_i=s_j$, $s_l$ is always unchanged compared to $\pi$ for $l\ne i,j$, and $V^{\pi'}_k(g)=V^{\pi}_k(g)$. The lemma follows by repeated application of this argument.

First, note that since $N_i(g)=N_j(g)$, for any $s\in S^\pi$ we have $A^{\pi}_i(g,s)=A^{\pi}_j(g,s)$. 
Hence, $i$ and $j$ have the same posterior belief; i.e., for any $\omega\in\Omega$ and any $s\in S^{\pi}$, it holds that $\lambda^{s,g}_i(\omega)=\lambda^{s,g}_j(\omega)$. Let $\vert S^{\pi}_i\times S^{\pi}_j\vert=c$.
Let $R=\left\{m^1,\ldots,m^c\right\}$ be a set of distinct messages.  
Define a bijection $\phi:S^{\pi}_i\times S^{\pi}_j\rightarrow R$. 
That is, for any \emph{tuple} $(s_i,s_j),(t_i,t_j)\in S^{\pi}_i\times S^{\pi}_j$, it holds that $\phi(s_i,s_j)=\phi(t_i,t_j)$ iff $(s_i,s_j)=(t_i,t_j)$; each \emph{combination of messages} of $i$ and $j$ is mapped to a separate message in $R$.

Define $S'=\left\{s'\in S\vert\: s\in S^{\pi}, s'_{-ij}=s_{-ij}\text{, and }\phi(s_i,s_j)=s'_i=s'_j\in R\right\}$. In words, $S'$ consists of signals obtained by replacing the messages of $i$ and $j$ with identical messages in $R$ (for each distinct combination of messages) and leaving the the other receivers' messages unchanged, for each signal in $S^{\pi}$. 
Let $\tau:S^{\pi}\rightarrow S'$ be a bijection s.t. for any $s\in S^{\pi}$, we have $\tau (s)=s'$ if $\tau(s_i,s_j)=s'_i=s'_j$ and $s'_{-ij}=s_{-ij}$.

For every $s\in S^\pi$ and $\omega\in\Omega$, define $\pi'\left(\tau(s)\vert\omega\right)=\pi(s\vert\omega)$. Clearly, $\pi'$ is an experiment. Since the probabilities to send corresponding signals conditional on the state are the same under $\pi$ and $\pi'$, receivers $i$ and $j$ still have the same posterior belief under $\pi'$; i.e., for any $\omega\in\Omega$ and $s\in S^{\pi'}$, it holds that: $\lambda^{s,g}_i(\omega)=\lambda^{s,g}_j(\omega)$. 

Next, we show that for any $r\in N_i(g)$, $\omega\in\Omega$, and $s\in S^{\pi}$, we have $\lambda^{s,g}_r(\omega)=\lambda^{\tau(s),g}_r(\omega)$. 
Indeed, 
\vspace{-0.25cm}
\begin{align*}
\lambda^{s,g}_r(\omega)&=\frac{\sum_{t\in A^{\pi}_r(g,s)}\pi(t\vert\omega)\lambda^0(\omega)}{\sum_{\omega'\in\Omega}\sum_{t\in A^{\pi}_r(g,s)}\pi(t\vert\omega')\lambda^0(\omega')}=\frac{\sum_{t\in A^{\pi}_r(g,s)}\pi'(\tau(t)\vert\omega)\lambda^0(\omega)}{\sum_{\omega'\in\Omega}\sum_{t\in A^{\pi}_r(g,s)}\pi'(\tau(t)\vert\omega')\lambda^0(\omega')} \\
&=\frac{\sum_{t'\in A^{\pi'}_r(g,\tau(s))}\pi'(t'\vert\omega)\lambda^0(\omega)}{\sum_{\omega'\in\Omega}\sum_{t'\in A^{\pi'}_r(g,\tau(s))}\pi'(t'\vert\omega')\lambda^0(\omega')}=\lambda^{\tau(s),g}_r(\omega).
\vspace{-0.25cm}
\end{align*}

\noindent Finally, any $r\notin N_i(g)$ has the same posterior belief under $\pi$ and $\pi'$, as it is not affected by the transformation. 
Hence, $V^{\pi'}_k(g)=V^{\pi}_k(g)$. 
\end{proof}

\noindent\textbf{Lemma \ref{lem:circle} (restated).}
\emph{Let $g\in G(N)$ be a circle and $n>3$. 
Then there exists $\pi\in\Pi$ s.t. $V^{\pi}_k(g)=V^n_k$.}
\begin{proof}
Let $\widetilde V=(\lambda^0(X)/\lambda^0(Y))\cdot(n/k)$.
For all $i\in N$, define $T_i=\{i\:(\textrm{mod } n),(i+1)\:(\textrm{mod } n),\ldots,(i+n-1-k)\:(\textrm{mod } n)\}$.
For each $T_i$, define $s^{T_i}\in S$ s.t. $s^{T_i}_j=x$ for all $j\in T_i$ and $s^{T_i}_\ell=y$ for all $\ell\in N\setminus T_i$. 
Define $\bar x$ and $\bar y$ with $\bar x_i=x$ and $\bar y_i=y$ for all $i\in N$. 
Define further $\pi\in\Pi$ by setting $\pi(\bar x\vert X)=1$, $\pi\left(s^{T_i}\vert Y\right)=\widetilde V/n$ for all $i\in N$, and $\pi(\bar y\vert Y)=1-\widetilde V$.
That is, $\pi$ sends $x$ to all the receivers in state $X$ with probability $1$, and in state $Y$ either sends $x$ to exactly $k$ consecutively-indexed receivers picked uniformly at random or sends $y$ to all the receivers. Note that $\pi$ is a probability distribution, and the marginal posterior of each receiver upon observing $x$ is uniform.
This completes the proof, since $V^{\pi}_k(g)=\lambda^0(X)\cdot 1+\lambda^0(Y)\cdot \widetilde V=\frac{n+k}{k}\lambda^0(X)=V^n_k$.
\end{proof}
\medskip

\noindent\textbf{Proposition \ref{pro:pairs} (restated).}
\emph{Let $g\in G(N)$ (with $n>3$) be s.t. for all $i\in N$, it holds that $\delta_i\leq 1$. 
Let $\pi\in\Pi$ be optimal on $g$. 
Then either
\begin{inparaenum}[(i)]
\item $V^\pi_k(g)=V^n_k$ or
\item there exists $g'\supsetneq g$ and $\pi'\in\Pi$ s.t. $V^{\pi'}_k(g')>V^\pi_k(g)$. 
\end{inparaenum}}
\begin{proof}

\noindent Let $u_1-v_1$, $u_2-v_2$,..., $u_{\ell}-v_{\ell}$ be all the pairs of connected nodes in $g$ and let $w_1,\ldots,w_{n-2\ell}$ be the isolated nodes. Let us add the links $v_1-u_2$, $v_2-u_3$,..., $v_{\ell-1}-u_{\ell}$, $v_{\ell}-w_1$, $w_1-w_2$,..., $w_{n-2\ell-1}-w_{n-2\ell}$, $w_{n-2\ell}-u_1$. Then the network becomes a circle. From Lemma~\ref{lem:circle}, the sender's value is $V_k^n$. Therefore, either the sender's value on $g$ is also $V_k^n$, or we found a beneficial extension for the sender.
\end{proof}

\begin{remark}
To establish the existence of a network that does not achieve $V^n_k$, consider the case $\delta_i=1$ for each $i\in N$.
That is, the network consists of pairs, which implies that $n=2m$ for $m\in\mathbb N$. Let $k=2\ell+1$. By Lemma \ref{lem:opt} and symmetry, $V^\pi_k(g)<V^n_k$.
\end{remark}

\begin{proof}[Proof of Theorem \ref{thm:genstar}]
Let $r$ be the \emph{center} of $C$ -- i.e., $r$ is connected to all the other nodes in $C$. We first show that the sender's value is strictly below $V_k^n$. Indeed, assume it is not true. We have $V_k^n=(n/k+1)\lambda^0(X)<1$. By Lemma~\ref{lem:opt}, whenever the state is $Y$, at least $n-k$ receivers have posterior $\lambda^{s,g}_i(X)=0$ -- i.e., they know the true state. Since $C$ contains at least $k+1$ nodes and $r$ observes all the information available to the receivers in $C$, we deduce that whenever the state is $Y$, $r$ knows the true state with probability $1$. Hence, when analyzing the sender's expected utility conditional on the state being $Y$, one can ignore the receiver $r$. By Lemma~\ref{lem:opt}, the sender's value is at most $V_k^{n-1}=\left(\frac{n-1}{k}+1\right)\lambda^0(X)<V_k^n$, a contradiction. We shall now prove that there exists an extension in which there are no information-dominating pairs of receivers. By Lemma \ref{cor:info}, it implies that the sender can achieve the value of $V_k^n$.

Consider first the case $\ell=1$ (see Figure~\ref{fig:T1_2}). Let $P:=C\setminus\{r\}$. Then either $(i)$ $\vert P\vert=\vert N\setminus C\vert$ or $(ii)$ $\vert P\vert>\vert N\setminus C\vert$. 
In case $(i)$, consider $g'\supsetneq g$ s.t. each $i\in N\setminus C$ is connected to a unique node in $P$, and all the nodes in $N\setminus C$ are connected to each other.\footnote{Some nodes in $N\setminus C$ might be connected to each other already in $g$, but it does not affect us.} Then there are no information-dominating pairs in $g'$, as desired. In case $(ii)$, fix $j\in N\setminus C$. Consider $g'\supsetneq g$ s.t. each $i\in N\setminus (C\cup\{j\})$ is connected to a unique node in $P$, $j$ is connected to all the nodes in $P$ that are not connected to any node in $N\setminus C$, and all the nodes in $N\setminus C$ are connected to each other. Then there are no information-dominating pairs in $g'$.\footnote{It is clear when $\vert N\setminus C\vert>1$, and it also holds when $\vert N\setminus C\vert=1$, since in the latter case $N\setminus C=\{j\}$.}

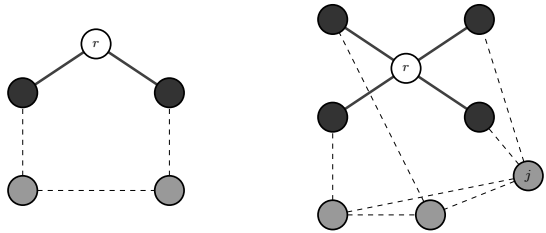
\begin{figure}[h]
\centering
\scalebox{.65}{
\begin{minipage}{0.45\textwidth}
\centering
\begin{tikzpicture}
\Vertex[color=black,opacity=0.8,x=0,y=0]{1}
\Vertex[color=black,opacity=0.8,x=3,y=0]{2}
\Vertex[color=white,opacity=0.8,x=1.5,y=1,label=$r$]{4}
\Vertex[color=gray,opacity=0.8,x=0,y=-2]{5}
\Vertex[color=gray,opacity=0.8,x=3,y=-2]{6}

\Edge(1)(4)
\Edge(2)(4)

\Edge[style={dashed, thin},color=black](5)(1)
\Edge[style={dashed, thin},color=black](6)(2)
\Edge[style={dashed, thin},color=black](5)(6)
\end{tikzpicture}
\end{minipage}%
\begin{minipage}{0.45\textwidth}
\centering
\begin{tikzpicture}[rectup/.style={rectangle, draw,dashed, font=\Large,minimum height=8mm,minimum width=40mm},rectdown/.style={rectangle, dotted, thick, draw, font=\Large,minimum height=8mm,minimum width=30mm}]

\Vertex[color=black,opacity=0.8,x=0,y=0]{1}
\Vertex[color=black,opacity=0.8,x=0,y=2]{2}
\Vertex[color=black,opacity=0.8,x=3,y=2]{3}
\Vertex[color=black,opacity=0.8,x=3,y=0]{7}
\Vertex[color=white,opacity=0.8,x=1.5,y=1,label=$r$]{4}

\Vertex[color=gray,opacity=0.8,x=0,y=-2]{5}
\Vertex[color=gray,opacity=0.8,x=2,y=-2]{6}
\Vertex[color=gray,opacity=0.8,x=4,y=-1.2,label=$j$]{8}

\Edge(1)(4)
\Edge(2)(4)
\Edge(3)(4)
\Edge(7)(4)

\Edge[style={dashed, thin},color=black](5)(1)
\Edge[style={dashed, thin},color=black](6)(2)
\Edge[style={dashed, thin},color=black](5)(6)
\Edge[style={dashed, thin},color=black](5)(8)
\Edge[style={dashed, thin},color=black](8)(6)
\Edge[style={dashed, thin},color=black](8)(3)
\Edge[style={dashed, thin},color=black](8)(7)

\end{tikzpicture}
\end{minipage}}
\caption{The extension of $g$ with $\ell= 1$ when $\vert P\vert=\vert N\setminus C\vert$ (left) and  $\vert P\vert>\vert N\setminus C\vert$ (right). The nodes in $P$ are in black, while the nodes in $N\setminus C$ are in gray; the original network is given by the solid lines, while its extension is given by the dashed lines.}
\label{fig:T1_2}
\end{figure}

Assume now $\ell>1$. Let us pick different nodes $v_1,\ldots,v_\ell$ in $N\setminus C$. Let us add an edge between each $v_i$ ($1\le i\le \ell$) and every node in $C$ having depth $i$ in the stellar component (see Figure~\ref{fig:T1} for illustration).

Consider now the nodes in $N\setminus(C\cup\{v_1,\ldots,v_\ell\})$ (if they exist). If $\vert N\setminus(C\cup\{v_1,\ldots,v_\ell\})\vert\ge 2$, let us connect the nodes of depth $1$ in $C$ to all the nodes of depth $2$ in $C$. In any case, connect all the nodes in $N\setminus(C\cup\{v_1,\ldots,v_\ell\})$ to each other and to $r,v_1,\ldots,v_\ell$, and connect each node in $N\setminus(C\cup\{v_1,\ldots,v_\ell\})$ to a different unique node in $C\setminus \{r\}$ (it is possible since $\vert C \vert > \vert N\setminus C \vert$). In the case $\vert N\setminus(C\cup\{v_1,\ldots,v_\ell\})\vert = 1$, make sure that the unique node in $N\setminus(C\cup\{v_1,\ldots,v_\ell\})$ is connected to a node of depth $2$ in $C$ (rather than an arbitrary node); see Figure~\ref{fig:T1_1} for illustration.

We claim that the above extension has no information-dominating pairs. Indeed, no node in $\{v_1,\ldots,v_\ell\}$ information-dominates another node in this set, as different $v_i$s are connected to nodes of different depths in $C$. Moreover, no node in $\{v_1,\ldots,v_\ell\}$ dominates a node in $C$ or $N\setminus(C\cup\{v_1,\ldots,v_\ell\})$ since the $v_i$s are not connected to~$r$.

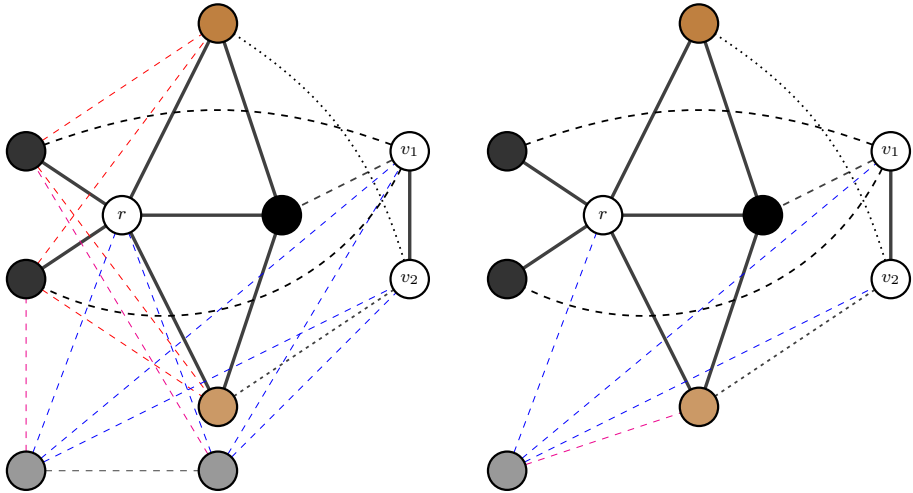
\begin{figure}[h]
\centering
\scalebox{.85}{
\begin{minipage}{0.45\textwidth}
\centering
\begin{tikzpicture}
\Vertex[color=black,opacity=0.8,x=0,y=0]{1}
\Vertex[color=black,opacity=0.8,x=0,y=2]{2} 
\Vertex[color=brown,x=3,y=4]{5} 
\Vertex[color=brown,opacity=0.8,x=3,y=-2]{3} 
\Vertex[color=white,opacity=0.8,x=1.5,y=1,label=$r$]{4}
\Vertex[color=black,x=4,y=1]{6}
\Vertex[color=white,x=6,y=2,label=$v_1$]{7}
\Vertex[color=white,x=6,y=0,label=$v_2$]{8}
\Vertex[color=gray,opacity=0.8,x=0,y=-3]{9} 
\Vertex[color=gray,opacity=0.8,x=3,y=-3]{10} 

\Edge(1)(4)
\Edge(4)(5)
\Edge(2)(4)
\Edge(3)(4)
\Edge(4)(6)
\Edge(3)(6)
\Edge(8)(7)
\Edge(6)(5)
\Edge[style={dashed, thick}](7)(6)
\Edge[style={dotted, thick}](8)(3)

\Edge[style={dashed, thin},color=red](1)(3)
\Edge[style={dashed, thin},color=red](1)(5)
\Edge[style={dashed, thin},color=red](2)(3)
\Edge[style={dashed, thin},color=red](2)(5)
\Edge[style={dashed, thin}](9)(10)
\Edge[style={dashed, thin},color=blue](9)(4)
\Edge[style={dashed, thin},color=blue](9)(7)
\Edge[style={dashed, thin},color=blue](9)(8)
\Edge[style={dashed, thin},color=blue](10)(4)
\Edge[style={dashed, thin},color=blue](10)(7)
\Edge[style={dashed, thin},color=blue](10)(8)
\Edge[style={dashed, thin},color=magenta](9)(1)
\Edge[style={dashed, thin},color=magenta](10)(2)

\path[-] (8) edge [bend right=20,style={dotted, thick}] node {} (5);
\path[-] (7) edge [bend right=20,style={dashed, thick}] node {} (2);
\path[-] (7) edge [bend left=45,style={dashed, thick}] node {} (1);
\end{tikzpicture}
\end{minipage}%
\begin{minipage}{0.54\textwidth}
\centering
\begin{tikzpicture}[rectup/.style={rectangle, draw,dashed, font=\Large,minimum height=8mm,minimum width=40mm},rectdown/.style={rectangle, dotted, thick, draw, font=\Large,minimum height=8mm,minimum width=30mm}]

\Vertex[color=black,opacity=0.8,x=0,y=0]{1}
\Vertex[color=black,opacity=0.8,x=0,y=2]{2} 
\Vertex[color=brown,x=3,y=4]{5}
\Vertex[color=brown,opacity=0.8,x=3,y=-2]{3} 
\Vertex[color=white,opacity=0.8,x=1.5,y=1,label=$r$]{4}
\Vertex[color=black,x=4,y=1]{6}
\Vertex[color=white,x=6,y=2,label=$v_1$]{7}
\Vertex[color=white,x=6,y=0,label=$v_2$]{8}
\Vertex[color=gray,opacity=0.8,x=0,y=-3]{9}

\Edge(1)(4)
\Edge(4)(5)
\Edge(2)(4)
\Edge(3)(4)
\Edge(4)(6)
\Edge(3)(6)
\Edge(8)(7)
\Edge(6)(5)
\Edge[style={dashed, thick}](7)(6)
\Edge[style={dotted, thick}](8)(3)

\Edge[style={dashed, thin},color=blue](9)(4)
\Edge[style={dashed, thin},color=blue](9)(7)
\Edge[style={dashed, thin},color=blue](9)(8)
\Edge[style={dashed, thin},color=magenta](9)(3)

\path[-] (8) edge [bend right=20,style={dotted, thick}] node {} (5);
\path[-] (7) edge [bend right=20,style={dashed, thick}] node {} (2);
\path[-] (7) edge [bend left=45,style={dashed, thick}] node {} (1);
\end{tikzpicture}
\end{minipage}}
\caption{The extension of the network $g$ with $\ell\ge 2$ when  $\vert N\setminus(C\cup\{v_1,\ldots,v_\ell\})\vert\ge 2$ (left) and  $\vert N\setminus(C\cup\{v_1,\ldots,v_\ell\})\vert=1$ (right). The nodes of depth $1$ in $C$ are given in black, the nodes of depth $2$ in 
$C$ -- in brown, the nodes in $N\setminus(C\cup\{v_1,\ldots,v_\ell\})$ -- in gray; the original network is represented by the solid lines, while the extension is represented by the dashed and the dotted lines.}
\label{fig:T1_1}
\end{figure}

Furthermore, no node in $N\setminus(C\cup\{v_1,\ldots,v_\ell\})$ can dominate another node in this set since the nodes in this set are connected to different unique nodes in $C\setminus \{r\}$; no node in $N\setminus(C\cup\{v_1,\ldots,v_\ell\})$ can dominate some $v_i$ since each $v_i$ is connected to at least two nodes in $C\setminus \{r\}$; and no node in $N\setminus(C\cup\{v_1,\ldots,v_\ell\})$ can dominate some node in $C$ as each node in $C$ observes the messages of at least three nodes in $C$ (including itself), except when $\vert N\setminus(C\cup\{v_1,\ldots,v_\ell\})\vert = 1$;\footnote{This is true even for nodes of depth $1$ since we connected them to all the nodes of depth $2$ if $\vert N\setminus(C\cup\{v_1,\ldots,v_\ell\})\vert\ge 2$.} and in the latter case, the unique node in $N\setminus(C\cup\{v_1,\ldots,v_\ell\})$ is connected only to a node of depth $2$ in $C$, which, indeed, observes the messages of at least three different nodes in $C$.

Finally, no node in $C$ dominates a node in $C\setminus\{r\}$ since nodes of different depths are connected to different nodes among $v_1,\ldots,v_\ell$, and $r$ is not connected to any $v_i$. $r$ is not dominated by nodes from $C\setminus\{r\}$ because if we added links between nodes in $C\setminus\{r\}$ then $r$ was connected to at least two nodes in $N\setminus(C\cup\{v_1,\ldots,v_\ell\})$. No node in $C$ dominates any $v_i$ since each $v_i$ is connected in $C$ exclusively to all the nodes of depth $i$ (of which there are at least two), which are not connected to each other. No node in $C\setminus\{r\}$ dominates a node in $N\setminus(C\cup\{v_1,\ldots,v_\ell\})$ since (unless the latter set is empty) either: (i) $\vert N\setminus(C\cup\{v_1,\ldots,v_\ell\}) \vert > 1$, and then each node in the set $N\setminus(C\cup\{v_1,\ldots,v_\ell\})$ is connected to at least one more node in this set; or (ii) $\vert N\setminus(C\cup\{v_1,\ldots,v_\ell\}) \vert = 1$, and then the unique node in $N\setminus(C\cup\{v_1,\ldots,v_\ell\})$ is not connected to any node from $C\setminus\{r\}$ that is connected to $v_1$; lastly, $r$ cannot dominate a node from $N\setminus(C\cup\{v_1,\ldots,v_\ell\})$ as $r$ is not connected to $v_1$. 
\end{proof}

\begin{proof}[Proof of Theorem \ref{thm:starlike}]
    We first claim that the sender cannot achieve expected utility of $V_n^k$ in $g$. Suppose, towards contradiction, that it was not the case. By Lemma~\ref{lem:opt}, under the sender's optimal experiment, every receiver has posterior either $\lambda^{s,g}_i(X)=0$ or $\lambda^{s,g}_i(X)=1/2$. If follows from Bayes-plausibility~\citep{kamenica2011bayesian} that all the receivers have the zero posterior $\lambda^{s,g}_i(X)=0$ with the same probability.
    
    Since the zero posterior is fully informative, and the central node of each component observes the posteriors of the peripheral nodes in that component -- if at least one node in a component has the zero posterior probability for the state to be $X$, then the central node of that component also has the zero posterior. Thus, the central node of each component has a weakly higher probability of having a zero posterior than each peripheral node in that component. To have equality, the probability that the central node of some component has the zero posterior and at least one node in the same component has a nonzero posterior must be zero. We know from Lemma~\ref{lem:opt} that with a positive probability, exactly $k$ nodes have nonzero posterior, which is not possible since there is no subset of the components with union of size $k$, a contradiction.

    It remains to prove that there is an extension of $g$ in which the sender can get $V_k^n$. We shall show that one can eliminate all the information-dominating pairs by adding links to $g$, which suffices by Lemma \ref{cor:info}. Sort the components in $g$ in a weakly decreasing order according to their size, and index the nodes in each component arbitrarily. Move iteratively over all the nodes according to the above ordering. During this traversal, whenever we reach a node that is only connected to nodes in its own component, connect it to the first node in the ordering which: (i) belongs to another component and (ii) has not been connected to any node outside its component yet.

    This process might get stuck. Suppose that it happens, and let $C$ be the component s.t.~the process got stuck while traversing it; if the process got stuck just after finishing traversing a component, let $C$ be the following component. Note that all nodes in the components after $C$ are already connected to nodes from other components. Moreover, $C$ must be the last component -- otherwise, all its nodes must already be connected to nodes from previous components, unless $C$ is the first component; but $C$ cannot be the first component, as it contains at most half of the~nodes.
    
    To complete the process, consider the nodes in $C$ that are not yet connected to any nodes in other components. Connect each of them to a node in the component $C'$ immediately preceding $C$ in the ordering s.t.~no two nodes in $C$ are connected to the same node in $C'$ (regardless of whether the links are new or existed before the process got stuck); it is possible since there are weakly more nodes in $C'$ than in $C$. We claim that in the resultant network there are no information-dominating pairs. 

    Indeed, nodes from the same component do not information-dominate each other since each of them is connected to a different node in another component. Peripheral nodes from different components also do not dominate each other. Indeed, consider two such peripheral nodes $u,v$. They are connected to different central nodes, unless, possibly, $u$ got connected during the process to a central node $c$ that had in its component $v$ (or vice versa). However, in this case, either: (i) $u$ precedes both $v$ and $c$ in the ordering, and when the process reached $v$ it could not connect it to $u$ (regardless of whether $v$ belongs to the last component in the ordering or not); or (ii) $c$ and $v$ precede $u$ in the ordering, and if $c$ got connected to $u$ then $v$ could not get connected to $u$. Neither way one of $u,v$ might dominate the other one.
    
    A peripheral node does not dominate any central nodes of other components. Indeed, it follows from the previous paragraph that if a peripheral node $u$ became connected to a central node $c$ of another component, then $u$ cannot be connected to any peripheral node in the component of $c$. Central nodes from different components do not dominate each other as if they became neighbors during the process, none of them could be connected to any of the peripheral nodes in the component of the other one. Finally, a central node $c$ cannot dominate a peripheral node $v$ from another component, since it could only happen if $c$ became a neighbor of both $v$ and the central node in the component of $v$, which is impossible.
\end{proof}

\begin{proof}[Proof of Proposition~\ref{pro:const}]
First notice that it follows from a similar argument as in the proof of Theorem~\ref{thm:genstar} that the sender cannot achieve $V^n_k$.
Let $\mu=\max\{\ell,\vert M\vert\}$ and take distinct $u,v_1,\ldots v_\mu\in N\setminus C$. Set $S:=M\cup\{u\}$. Take any $c\in M$ and consider the subnetwork $g^c_S$. Assume $\vert M\vert \ge 2$ (otherwise, we are done by Theorem~\ref{thm:genstar}).

Apply the same construction of the extension as in the proof of Theorem~\ref{thm:genstar}, using $v_1,\ldots,v_{\ell}$ as the corresponding nodes in the proof of Theorem~\ref{thm:genstar} if $\ell>1$, and using $v_1$ as $j$ in the proof of Theorem~\ref{thm:genstar} if $\ell=1$ and $\vert P\vert>\vert N\setminus C\vert$. Note that this construction does not break all information-dominating pairs, as: (i) the centers (i.e., the members of $M$) still information-dominate each other (possibly apart from $c$ that no longer needs to be dominated); and (ii) $u$ might still be dominated.

Assume first $\ell>1$. Let us further connect each member of $M\setminus\{c\}$ to a unique node among $v_2,\ldots,v_\mu$ and connect $u$ to $c,v_1,\ldots,v_{\mu}$, and to arbitrary $w_1,w_2\in C\setminus M$ of depths $1$ and $2$, respectively. Our modification of the extension in Theorem~\ref{thm:genstar} proof cannot make any node in $N$ dominated by nodes outside of $M\cup\{u,v_1,\ldots,v_{\mu},w_1,w_2\}$, since we do not affect the neighborhoods of nodes outside of this set. Moreover, the members of $M$ no longer information-dominate each other. They are also not dominated by any of the nodes $u,v_1,\ldots v_\mu$, since these nodes are connected to at most one member of $M$; and they are not dominated by $w_1,w_2$, as these two nodes are not connected to other nodes in $C$ of their own depth.

Moreover, $w_1,w_2$ do not dominate each other as they are connected to different $v_i$s. They are not dominated by $v_3,\ldots,v_{\mu}$ as they are not connected to them. They further cannot be dominated by members of $M$ as they are connected to $v_1$ or $v_2$, and cannot be dominated by $u,v_1,v_2$ as they are connected to all the members of~$M$.

The $v_i$s are not dominated by members of $M\setminus\{c\}$ since they are connected to $u$; they are not dominated by $c$ since they are not connected to it; they are not dominated by $w_1,w_2$ as $v_i$ can only be connected to $w_i$, and in this case $v_i$ is also connected to at least one more node of depth $i$; and they are not dominated by each other since they are connected to different nodes in $C$.\footnote{If $\ell\ge \vert M\vert$, they are connected to different members of $C\setminus M$; otherwise, $v_2,\ldots,v_\mu$ are connected to different members of $M$, and $v_1$ is the only one among the $v_i$s connected to a depth $1$ node in~$C$.}

Similarly, we get that $u$ is not dominated by members of $M\setminus\{c\}$ as they are not connected to $u$; it is not dominated by $c,v_1,\ldots,v_{\mu}$ since it is connected to all of them. Since we do not know that $u$ was not dominated after the construction from the proof of Theorem~\ref{thm:genstar}, we still need to make sure that it is not dominated by nodes outside of $M\cup\{u,v_1,\ldots,v_{\mu}\}$. Luckily, this is indeed the case, since the nodes in $C$ are connected to at most one of $v_1,\ldots,v_{\mu}$, and the nodes in $N\setminus (C\cup M\cup\{u,v_1,\ldots,v_{\mu}\})$ are connected to at most one node in $C\setminus M$.

Moreover, any node $w\notin M\cup\{u,v_1,\ldots,v_{\mu},w_1,w_2\}$ remains undominated. Indeed, $w$ is not connected to members of $M\setminus\{c\}$, is connected to at most one of $w_1,w_2$ (unlike $u,c$), and is connected to $c$ (unlike $v_1,\ldots,v_{\mu}$). Therefore, the resulting network does not have any information-dominating pairs. It follows from Lemma \ref{cor:info} that the sender can achieve $V^n_k$, as desired.

It remains to consider the case $\ell=1$. As no node in $g\setminus C$ is connected to a member of $M$ after the extension from Theorem~\ref{thm:genstar} proof, one can just further connect $u,v_2,\ldots,v_{\mu=\vert M\vert + 1}$ to another member of $M$ and connect every pair among $u,v_1,\ldots,v_{\mu}$ to break all the information-dominating pairs without creating new ones.
\end{proof}

\begin{proof}[Proof of Proposition \ref{pro:cliques}]
    We first claim that the sender's value for $g$ is strictly smaller than $V_k^n=(n/k+1)\lambda^0(X)<1$. By Lemma~\ref{lem:opt}, there is a positive probability that exactly $k$ receivers have posterior $\lambda^{s,g}_i(X)\ge 1/2$ and the remaining $n-k$ receivers have posterior $\lambda^{s,g}_i(X)=0$. Since the receivers belonging to the same cluster observe the posteriors of each other, and the zero posterior $\lambda^{s,g}_i(X)=0$ truthfully reveals the state -- we have that for any signal realization and every cluster, either all of the receivers in it have the zero posterior or none of the receivers in it have the zero posterior. Therefore, $p\mid n-k$. However, we also have $p\mid pq=n$, implying $p\mid k$. Since $n/2<k<n$, we must have that $n,k$ are relatively prime, a contradiction.

    It remains to show that there is an extension of $g$ in which the sender can achieve $V_k^n$. Let us index the receivers by pairs $(i,j)\in [p]\times [q]$, where $j$ is the index of the cluster and $i$ is the index of the receiver inside its cluster. Let us add, for each $i\in [p]$, the following $q$ edges to $g$: $(i,1)-(i,2),\; (i,2)-(i,3),\; \ldots,\; (i,q-1)-(i,q),\; (i,q)-(i,1)$. That is -- we add circles connecting all the receivers having the same index within their clusters. We claim that there are no information-dominating pairs of receivers in the resultant network, which completes the proof by Lemma \ref{cor:info}. Indeed, consider two receivers $(i_1,j_1)\neq (i_2,j_2)$. We claim that $(i_1,j_1)$ does not information-dominate $(i_2,j_2)$. Indeed, if $j_1=j_2$ (i.e., the receivers belong to the same cluster), then $(i_2,j_2)$ observes the message sent to $(i_2, j_2 + 1 \mod q)$, while $(i_1,j_1)$ does not. If $i_1=i_2$ (i.e., the receivers have the same index within their clusters), then  $(i_2,j_2)$ observes the message sent to $(i_2 + 1 \mod p, j_2)$, while $(i_1,j_1)$ does not. Finally, if $i_1\neq i_2$ and $j_1\neq j_2$, then $(i_1,j_1)$ does not observe the message sent to $(i_2,j_2)$.
\end{proof}

\end{appendices}

{\bf Acknowledgments.} We are grateful to Vyacheslav Arbuzov, Itai Arieli, Yakov Babichenko, Dirk Bergemann, Matteo Bizzarri, Manuel Foerster, Arda Gitmez, Kevin Hasker, Jean-Jacques Herings, András Kálecz-Simon, Dominik Karos, Frederic Koessler, László Kóczy, Friederike Mengel, Ronald Peeters, Mikl\'{o}s Pint\'{e}r, Agnieszka Rusinowska, Dov Samet, Sudipta Sarangi, Maxim Senkov, Inbal Talgam-Cohen, Frank Thuijsman, Péter Vida, and Kemal Y{\i}ld{\i}z for their valuable comments and suggestions. We thank the participants of the HUN-REN Center for Economic and Regional Studies KTI seminar and the European Meeting on Game Theory 2025~(SING20).

\bibliography{sources}

\begin{thebibliography}{}

\bibitem[\protect\citeauthoryear{Ali and Miller}{Ali and Miller}{2016}]{ali2016ostracism}
Ali, S.~N. and D.~A. Miller (2016).
\newblock Ostracism and forgiveness.
\newblock {\em American Economic Review\/}~{\em 106\/}(8), 2329--2348.

\bibitem[\protect\citeauthoryear{Alonso and C{\^a}mara}{Alonso and C{\^a}mara}{2016}]{alonso2016persuading}
Alonso, R. and O.~C{\^a}mara (2016).
\newblock Persuading voters.
\newblock {\em American Economic Review\/}~{\em 106\/}(11), 3590--3605.

\bibitem[\protect\citeauthoryear{Anderson and Holt}{Anderson and Holt}{1997}]{anderson1997information}
Anderson, L.~R. and C.~A. Holt (1997).
\newblock Information cascades in the laboratory.
\newblock {\em The American economic review\/}, 847--862.

\bibitem[\protect\citeauthoryear{Angeletos and Pavan}{Angeletos and Pavan}{2004}]{angeletos2004transparency}
Angeletos, G.-M. and A.~Pavan (2004).
\newblock Transparency of information and coordination in economies with investment complementarities.
\newblock {\em American Economic Review\/}~{\em 94\/}(2), 91--98.

\bibitem[\protect\citeauthoryear{Anunrojwong and Sothanaphan}{Anunrojwong and Sothanaphan}{2018}]{anunrojwong2018naive}
Anunrojwong, J. and N.~Sothanaphan (2018).
\newblock {Naive Bayesian learning in social networks}.
\newblock In {\em Proceedings of the 2018 ACM Conference on Economics and Computation}, pp.\  619--636.

\bibitem[\protect\citeauthoryear{Arieli and Babichenko}{Arieli and Babichenko}{2019}]{Arieli19}
Arieli, I. and Y.~Babichenko (2019).
\newblock {Private Bayesian persuasion}.
\newblock {\em Journal of Economic Theory\/}~{\em 182}, 185--217.

\bibitem[\protect\citeauthoryear{Athanassopoulos}{Athanassopoulos}{2000}]{athanassopoulos2000customer}
Athanassopoulos, A.~D. (2000).
\newblock Customer satisfaction cues to support market segmentation and explain switching behavior.
\newblock {\em Journal of business research\/}~{\em 47\/}(3), 191--207.

\bibitem[\protect\citeauthoryear{Babichenko, Talgam-Cohen, Xu, and Zabarnyi}{Babichenko et~al.}{2021}]{Babichenko21}
Babichenko, Y., I.~Talgam-Cohen, H.~Xu, and K.~Zabarnyi (2021).
\newblock {Multi-Channel Bayesian Persuasion}.
\newblock {\em arXiv preprint arXiv:2111.09789\/}.

\bibitem[\protect\citeauthoryear{Backstrom and Leskovec}{Backstrom and Leskovec}{2011}]{backstrom2011supervised}
Backstrom, L. and J.~Leskovec (2011).
\newblock Supervised random walks: predicting and recommending links in social networks.
\newblock In {\em Proceedings of the fourth ACM international conference on Web search and data mining}, pp.\  635--644.

\bibitem[\protect\citeauthoryear{Bala and Goyal}{Bala and Goyal}{2000}]{bala2000noncooperative}
Bala, V. and S.~Goyal (2000).
\newblock A noncooperative model of network formation.
\newblock {\em Econometrica\/}~{\em 68\/}(5), 1181--1229.

\bibitem[\protect\citeauthoryear{Bardhi and Guo}{Bardhi and Guo}{2018}]{Bardhi18}
Bardhi, A. and Y.~Guo (2018).
\newblock Modes of persuasion toward unanimous consent.
\newblock {\em Theoretical Economics\/}~{\em 13\/}(3), 1111--1149.

\bibitem[\protect\citeauthoryear{Benevento, Aloini, Roma, and Bellino}{Benevento et~al.}{2025}]{benevento2025impact}
Benevento, E., D.~Aloini, P.~Roma, and D.~Bellino (2025).
\newblock {The impact of influencers on brand social network growth: Insights from new product launch events on Twitter}.
\newblock {\em Journal of Business Research\/}~{\em 189}, 115123.

\bibitem[\protect\citeauthoryear{Bikhchandani, Hirshleifer, and Welch}{Bikhchandani et~al.}{1992}]{bikhchandani1992theory}
Bikhchandani, S., D.~Hirshleifer, and I.~Welch (1992).
\newblock A theory of fads, fashion, custom, and cultural change as informational cascades.
\newblock {\em Journal of political Economy\/}~{\em 100\/}(5), 992--1026.

\bibitem[\protect\citeauthoryear{Bresman, Birkinshaw, and Nobel}{Bresman et~al.}{1999}]{bresman1999knowledge}
Bresman, H., J.~Birkinshaw, and R.~Nobel (1999).
\newblock Knowledge transfer in international acquisitions.
\newblock {\em Journal of international business studies\/}~{\em 30\/}(3), 439--462.

\bibitem[\protect\citeauthoryear{Broderick, Greenley, and Mueller}{Broderick et~al.}{2007}]{broderick2007behavioural}
Broderick, A.~J., G.~E. Greenley, and R.~D. Mueller (2007).
\newblock The behavioural homogeneity evaluation framework: Multi-level evaluations of consumer involvement in international segmentation.
\newblock {\em Journal of International Business Studies\/}~{\em 38\/}(5), 746--763.

\bibitem[\protect\citeauthoryear{Burt}{Burt}{1992}]{Burt92}
Burt, R.~S. (1992).
\newblock {\em Structural holes}.
\newblock Harvard University Press.

\bibitem[\protect\citeauthoryear{Burt}{Burt}{2004}]{Burt04}
Burt, R.~S. (2004).
\newblock Structural holes and good ideas.
\newblock {\em American Journal of Sociology\/}~{\em 110\/}(2), 349 – 399.

\bibitem[\protect\citeauthoryear{Candogan}{Candogan}{2019}]{Candogan19}
Candogan, O. (2019).
\newblock Persuasion in networks: Public signals and k-cores.
\newblock New York, NY, USA. Association for Computing Machinery.

\bibitem[\protect\citeauthoryear{Candogan, Guo, and Xu}{Candogan et~al.}{2020}]{Candogan2020}
Candogan, O., Y.~Guo, and H.~Xu (2020).
\newblock On information design with spillovers.
\newblock {\em Available at SSRN 3537289\/}.

\bibitem[\protect\citeauthoryear{Chan, Gupta, Li, and Wang}{Chan et~al.}{2019}]{Chan19}
Chan, J., S.~Gupta, F.~Li, and Y.~Wang (2019).
\newblock Pivotal persuasion.
\newblock {\em Journal of Economic Theory\/}~{\em 180}, 178--202.

\bibitem[\protect\citeauthoryear{Church and Gandal}{Church and Gandal}{1992}]{church1992network}
Church, J. and N.~Gandal (1992).
\newblock Network effects, software provision, and standardization.
\newblock {\em The journal of industrial economics\/}, 85--103.

\bibitem[\protect\citeauthoryear{Colla and Mele}{Colla and Mele}{2010}]{colla2010information}
Colla, P. and A.~Mele (2010).
\newblock Information linkages and correlated trading.
\newblock {\em The Review of Financial Studies\/}~{\em 23\/}(1), 203--246.

\bibitem[\protect\citeauthoryear{Colman and Rouzies}{Colman and Rouzies}{2019}]{colman2019postacquisition}
Colman, H.~L. and A.~Rouzies (2019).
\newblock Postacquisition boundary spanning: A relational perspective on integration.
\newblock {\em Journal of Management\/}~{\em 45\/}(5), 2225--2253.

\bibitem[\protect\citeauthoryear{Cooperman}{Cooperman}{2024}]{cooperman2024bloc}
Cooperman, A.~D. (2024).
\newblock Bloc voting for electoral accountability.
\newblock {\em American Political Science Review\/}~{\em 118\/}(3), 1222--1239.

\bibitem[\protect\citeauthoryear{Corten and Buskens}{Corten and Buskens}{2010}]{corten2010co}
Corten, R. and V.~Buskens (2010).
\newblock {Co-evolution of conventions and networks: An experimental study}.
\newblock {\em Social Networks\/}~{\em 32\/}(1), 4--15.

\bibitem[\protect\citeauthoryear{Currarini, Ursino, and Chand}{Currarini et~al.}{2020}]{currarini2020strategic}
Currarini, S., G.~Ursino, and A.~Chand (2020).
\newblock Strategic transmission of correlated information.
\newblock {\em The Economic Journal\/}~{\em 130\/}(631), 2175--2206.

\bibitem[\protect\citeauthoryear{de~Fine~Licht}{de~Fine~Licht}{2014}]{de2014transparency}
de~Fine~Licht, J. (2014).
\newblock Transparency actually: how transparency affects public perceptions of political decision-making.
\newblock {\em European Political Science Review\/}~{\em 6\/}(2), 309--330.

\bibitem[\protect\citeauthoryear{Dong, Tang, Chawla, Lou, Yang, and Wang}{Dong et~al.}{2015}]{dong2015inferring}
Dong, Y., J.~Tang, N.~V. Chawla, T.~Lou, Y.~Yang, and B.~Wang (2015).
\newblock Inferring social status and rich club effects in enterprise communication networks.
\newblock {\em PloS one\/}~{\em 10\/}(3), e0119446.

\bibitem[\protect\citeauthoryear{Drago, Mengel, and Traxler}{Drago et~al.}{2020}]{Drago20}
Drago, F., F.~Mengel, and C.~Traxler (2020).
\newblock Compliance behavior in networks: Evidence from a field experiment.
\newblock {\em American Economic Journal: Applied Economics\/}~{\em 12\/}(2), 96–133.

\bibitem[\protect\citeauthoryear{Egorov and Sonin}{Egorov and Sonin}{2020}]{Egorov19}
Egorov, G. and K.~Sonin (2020).
\newblock Persuasion on networks.
\newblock Technical report, National Bureau of Economic Research.

\bibitem[\protect\citeauthoryear{Eguia}{Eguia}{2011}]{eguia2011voting}
Eguia, J.~X. (2011).
\newblock Voting blocs, party discipline and party formation.
\newblock {\em Games and economic behavior\/}~{\em 73\/}(1), 111--135.

\bibitem[\protect\citeauthoryear{Galeotti and Goyal}{Galeotti and Goyal}{2010}]{galeotti2010law}
Galeotti, A. and S.~Goyal (2010).
\newblock The law of the few.
\newblock {\em American Economic Review\/}~{\em 100\/}(4), 1468--92.

\bibitem[\protect\citeauthoryear{Galperti and Perego}{Galperti and Perego}{2025}]{Galperti23}
Galperti, S. and J.~Perego (2025).
\newblock Games with information constraints: seeds and spillovers.
\newblock {\em Theoretical Economics\/}~{\em 20}, 667–711.

\bibitem[\protect\citeauthoryear{Gatewood}{Gatewood}{1984}]{gatewood1984cooperation}
Gatewood, J.~B. (1984).
\newblock {Cooperation, competition, and synergy: information-sharing groups among Southeast Alaskan Salmon seiners}.
\newblock {\em American Ethnologist\/}~{\em 11\/}(2), 350--370.

\bibitem[\protect\citeauthoryear{Gavazza and Lizzeri}{Gavazza and Lizzeri}{2009}]{gavazza2009transparency}
Gavazza, A. and A.~Lizzeri (2009).
\newblock Transparency and economic policy.
\newblock {\em The Review of Economic Studies\/}~{\em 76\/}(3), 1023--1048.

\bibitem[\protect\citeauthoryear{Goeree, Riedl, and Ule}{Goeree et~al.}{2009}]{goeree2009search}
Goeree, J.~K., A.~Riedl, and A.~Ule (2009).
\newblock In search of stars: Network formation among heterogeneous agents.
\newblock {\em Games and Economic Behavior\/}~{\em 67\/}(2), 445--466.

\bibitem[\protect\citeauthoryear{Gormley and Murphy}{Gormley and Murphy}{2008}]{gormley2008exploring}
Gormley, I.~C. and T.~B. Murphy (2008).
\newblock Exploring voting blocs within the irish electorate: A mixture modeling approach.
\newblock {\em Journal of the American Statistical Association\/}~{\em 103\/}(483), 1014--1027.

\bibitem[\protect\citeauthoryear{Grimmer, Marble, and Tanigawa-Lau}{Grimmer et~al.}{2025}]{grimmer2025measuring}
Grimmer, J., W.~Marble, and C.~Tanigawa-Lau (2025).
\newblock Measuring the contribution of voting blocs to election outcomes.
\newblock {\em The Journal of Politics\/}~{\em 87\/}(3), 905--921.

\bibitem[\protect\citeauthoryear{Gupta, Goel, Lin, Sharma, Wang, and Zadeh}{Gupta et~al.}{2013}]{gupta2013wtf}
Gupta, P., A.~Goel, J.~Lin, A.~Sharma, D.~Wang, and R.~Zadeh (2013).
\newblock Wtf: The who to follow service at twitter.
\newblock In {\em Proceedings of the 22nd international conference on World Wide Web}, pp.\  505--514.

\bibitem[\protect\citeauthoryear{Hansen}{Hansen}{1999}]{hansen1999search}
Hansen, M.~T. (1999).
\newblock The search-transfer problem: The role of weak ties in sharing knowledge across organization subunits.
\newblock {\em Administrative science quarterly\/}~{\em 44\/}(1), 82--111.

\bibitem[\protect\citeauthoryear{Jackson and Rogers}{Jackson and Rogers}{2007}]{Jackson07}
Jackson, M.~O. and B.~W. Rogers (2007).
\newblock Meeting strangers and friends of friends: How random are social networks?
\newblock {\em The American Economic Review\/}~{\em 97\/}(3), 890--915.

\bibitem[\protect\citeauthoryear{Jankowski, Br{\'o}dka, and Hamari}{Jankowski et~al.}{2016}]{jankowski2016picture}
Jankowski, J., P.~Br{\'o}dka, and J.~Hamari (2016).
\newblock A picture is worth a thousand words: an empirical study on the influence of content visibility on diffusion processes within a virtual world.
\newblock {\em Behaviour \& Information Technology\/}~{\em 35\/}(11), 926--945.

\bibitem[\protect\citeauthoryear{Kamenica and Gentzkow}{Kamenica and Gentzkow}{2011}]{kamenica2011bayesian}
Kamenica, E. and M.~Gentzkow (2011).
\newblock Bayesian persuasion.
\newblock {\em American Economic Review\/}~{\em 101\/}(6), 2590--2615.

\bibitem[\protect\citeauthoryear{Karamched, Stolarczyk, Kilpatrick, and Josic}{Karamched et~al.}{2020}]{karamched2020bayesian}
Karamched, B., S.~Stolarczyk, Z.~P. Kilpatrick, and K.~Josic (2020).
\newblock Bayesian evidence accumulation on social networks.
\newblock {\em SIAM journal on applied dynamical systems\/}~{\em 19\/}(3), 1884--1919.

\bibitem[\protect\citeauthoryear{Katona, Zubcsek, and Sarvary}{Katona et~al.}{2011}]{katona2011network}
Katona, Z., P.~P. Zubcsek, and M.~Sarvary (2011).
\newblock Network effects and personal influences: The diffusion of an online social network.
\newblock {\em Journal of marketing research\/}~{\em 48\/}(3), 425--443.

\bibitem[\protect\citeauthoryear{Katz and Shapiro}{Katz and Shapiro}{1994}]{katz1994systems}
Katz, M.~L. and C.~Shapiro (1994).
\newblock Systems competition and network effects.
\newblock {\em Journal of economic perspectives\/}~{\em 8\/}(2), 93--115.

\bibitem[\protect\citeauthoryear{Kerman and Tenev}{Kerman and Tenev}{2021}]{kerman2021persuading}
Kerman, T. and A.~P. Tenev (2021).
\newblock Persuading communicating voters.
\newblock {\em Available at SSRN 3765527\/}.

\bibitem[\protect\citeauthoryear{Kerman, Herings, and Karos}{Kerman et~al.}{2024}]{kerman2024persuading}
Kerman, T.~T., P.~J.-J. Herings, and D.~Karos (2024).
\newblock Persuading sincere and strategic voters.
\newblock {\em Journal of Public Economic Theory\/}~{\em 26\/}(1), e12671.

\bibitem[\protect\citeauthoryear{Kerman and Tenev}{Kerman and Tenev}{2025}]{kerman2024information}
Kerman, T.~T. and A.~P. Tenev (2025).
\newblock Information design for weighted voting.
\newblock {\em Economic Theory\/}~{\em 79}, 809--852.

\bibitem[\protect\citeauthoryear{Kerman, Tenev, and Tsodikovich}{Kerman et~al.}{2025}]{kerman2025bayesian}
Kerman, T.~T., A.~P. Tenev, and Y.~Tsodikovich (2025).
\newblock Bayesian persuasion in networks: Divisibility and network irrelevance.
\newblock {\em Available at SSRN 5131130\/}.

\bibitem[\protect\citeauthoryear{Levy}{Levy}{2007}]{levy2007decision}
Levy, G. (2007).
\newblock Decision making in committees: Transparency, reputation, and voting rules.
\newblock {\em American economic review\/}~{\em 97\/}(1), 150--168.

\bibitem[\protect\citeauthoryear{Liporace}{Liporace}{2021}]{liporace2021persuasion}
Liporace, M. (2021).
\newblock Persuasion in networks.
\newblock {\em Working Paper\/}.

\bibitem[\protect\citeauthoryear{Marmaros and Sacerdote}{Marmaros and Sacerdote}{2006}]{Marmaros06}
Marmaros, D. and B.~Sacerdote (2006, 02).
\newblock {How Do Friendships Form?}
\newblock {\em The Quarterly Journal of Economics\/}~{\em 121\/}(1), 79--119.

\bibitem[\protect\citeauthoryear{Mathevet, Perego, and Taneva}{Mathevet et~al.}{2020}]{Mathevet2020}
Mathevet, L., J.~Perego, and I.~Taneva (2020).
\newblock On information design in games.
\newblock {\em Journal of Political Economy\/}~{\em 128\/}(4), 1370--1404.

\bibitem[\protect\citeauthoryear{Mathiesen, Jamtveit, and Sneppen}{Mathiesen et~al.}{2010}]{mathiesen2010organizational}
Mathiesen, J., B.~Jamtveit, and K.~Sneppen (2010).
\newblock Organizational structure and communication networks in a university environment.
\newblock {\em Physical Review E—Statistical, Nonlinear, and Soft Matter Physics\/}~{\em 82\/}(1), 016104.

\bibitem[\protect\citeauthoryear{Mengel, Friedman, and Kov{\'a}{\v{r}}{\'\i}k}{Mengel et~al.}{2024}]{mengel2024influence}
Mengel, F., D.~Friedman, and J.~Kov{\'a}{\v{r}}{\'\i}k (2024).
\newblock Influence in social networks.
\newblock {\em Available at SSRN 5059047\/}.

\bibitem[\protect\citeauthoryear{Molavi and Jadbabaie}{Molavi and Jadbabaie}{2011}]{molavi2011aggregate}
Molavi, P. and A.~Jadbabaie (2011).
\newblock Aggregate observational distinguishability is necessary and sufficient for social learning.
\newblock In {\em 2011 50th IEEE Conference on Decision and Control and European Control Conference}, pp.\  2335--2340. IEEE.

\bibitem[\protect\citeauthoryear{Noam}{Noam}{2016}]{noam2016owns}
Noam, E.~M. (2016).
\newblock {\em {Who Owns the World's Media? Media Concentration and Ownership around the World}}.
\newblock Oxford University Press.

\bibitem[\protect\citeauthoryear{Noble and Jones}{Noble and Jones}{2006}]{noble2006role}
Noble, G. and R.~Jones (2006).
\newblock The role of boundary-spanning managers in the establishment of public-private partnerships.
\newblock {\em Public administration\/}~{\em 84\/}(4), 891--917.

\bibitem[\protect\citeauthoryear{Nora and Winter}{Nora and Winter}{2024}]{nora2024exploiting}
Nora, V. and E.~Winter (2024).
\newblock Exploiting social influence in networks.
\newblock {\em Theoretical Economics\/}~{\em 19\/}(1), 1--27.

\bibitem[\protect\citeauthoryear{Sun, Schram, and Sloof}{Sun et~al.}{2023}]{sun2023public}
Sun, J., A.~J. Schram, and R.~Sloof (2023).
\newblock Public persuasion in elections: Single-crossing property and the optimality of censorship.
\newblock {\em Available at SSRN 4028840\/}.

\bibitem[\protect\citeauthoryear{Taneva}{Taneva}{2019}]{Taneva19}
Taneva, I. (2019).
\newblock Information design.
\newblock {\em American Economic Journal: Microeconomics\/}~{\em 11\/}(4), 151--85.

\bibitem[\protect\citeauthoryear{Titova}{Titova}{2022}]{titova2022persuasion}
Titova, M. (2022).
\newblock Persuasion with verifiable information.
\newblock Technical report, Working~Paper.

\bibitem[\protect\citeauthoryear{Vizcarrondo}{Vizcarrondo}{2013}]{vizcarrondo2013measuring}
Vizcarrondo, T. (2013).
\newblock Measuring concentration of media ownership: 1976--2009.
\newblock {\em International Journal on Media Management\/}~{\em 15\/}(3), 177--195.

\bibitem[\protect\citeauthoryear{Wang}{Wang}{2013}]{Wang13}
Wang, Y. (2013).
\newblock Bayesian persuasion with multiple receivers.
\newblock {\em Available at SSRN: https://ssrn.com/abstract=2625399\/}.

\bibitem[\protect\citeauthoryear{Williams}{Williams}{2012}]{williams2012we}
Williams, P. (2012).
\newblock We are all boundary spanners now?
\newblock In {\em Collaboration in Public Policy and Practice}, pp.\  95--118. Policy Press.

\end{thebibliography}

\end{document}